\documentclass[twocolumn,english,letterpaper,superscriptaddress]{revtex4-1}
\usepackage[T1]{fontenc}
\usepackage[latin9]{inputenc}
\usepackage{amsmath}
\usepackage{amsthm}
\usepackage{amssymb}
\usepackage{color}
\usepackage{epsfig}
\usepackage{graphicx}
\usepackage[normalem]{ulem}

\makeatletter
\@ifundefined{textcolor}{}
{%
 \definecolor{BLACK}{gray}{0}
 \definecolor{WHITE}{gray}{1}
 \definecolor{RED}{rgb}{1,0,0}
 \definecolor{GREEN}{rgb}{0,1,0}
 \definecolor{BLUE}{rgb}{0,0,1}
 \definecolor{CYAN}{cmyk}{1,0,0,0}
 \definecolor{MAGENTA}{cmyk}{0,1,0,0}
 \definecolor{YELLOW}{cmyk}{0,0,1,0}
}

\theoremstyle{plain}
\newtheorem{thm}{\protect\theoremname}
  \theoremstyle{definition}
  \newtheorem{defn}[thm]{\protect\definitionname}
  \theoremstyle{plain}
  \newtheorem{cor}[thm]{\protect\corollaryname}
  \theoremstyle{plain}
  \newtheorem{fact}[thm]{\protect\factname}
  \theoremstyle{plain}
  \newtheorem{lem}[thm]{\protect\lemmaname}

\theoremstyle{plain}
  \newtheorem*{lem*}{\protect\lemmaname}	
\theoremstyle{plain}
\newtheorem*{thm*}{\protect\theoremname}
	
\makeatother

\usepackage{babel}
  \providecommand{\corollaryname}{Corollary}
  \providecommand{\definitionname}{Definition}
  \providecommand{\factname}{Fact}
  \providecommand{\lemmaname}{Lemma}
\providecommand{\theoremname}{Theorem}

\newcommand{\cA}{\mathcal{A}}
\newcommand{\cH}{\mathcal{H}}
\newcommand{\cN}{\mathcal{N}}

\newcommand{\id}{\mathbbm{1}}
\newcommand{\mathbbm}[1]{\text{\usefont{U}{bbm}{m}{n}#1}}
\newcommand{\ket}[1]{ | #1 \rangle}

\newcommand{\ketbra}[2]{|#1\rangle\langle#2|}

\newcommand{\be}{\begin{equation}}
\newcommand{\ee}{\end{equation}}
\newcommand{\bq}{\begin{eqnarray}}
\newcommand{\eq}{\end{eqnarray}}

\makeatother

\usepackage{babel}
\begin{document}

\title{Anyons are not energy eigenspaces of quantum double Hamiltonians}

\author{Anna K\'om\'ar}
\affiliation{Institute for Quantum Information and Matter and Walter Burke Institute for Theoretical Physics, California Institute of Technology, Pasadena, California 91125, USA}
\author{Olivier Landon-Cardinal}
\affiliation{Departement of Physics, McGill University, Montr\'eal, Qu\'ebec, Canada H3A 2T8}
\begin{abstract}

Kitaev's quantum double models, including the toric code, are canonical examples of quantum topological models on a 2D spin lattice. Their Hamiltonian defines the groundspace by imposing an energy penalty to any nontrivial flux or charge, but does not distinguish among those. We generalize this construction by introducing a novel family of Hamiltonians made of commuting four-body projectors that provide an intricate splitting of the Hilbert space by discriminating among non-trivial charges and fluxes. Our construction highlights that anyons are not in one-to-one correspondence with energy eigenspaces, a feature already present in Kitaev's construction. This discrepancy is due to the presence of local degrees of freedom in addition to topological ones on a lattice. 

\end{abstract}
\maketitle

\section{Introduction}

Interacting topological spin models are of interest in the field of condensed matter theory and quantum information due to their promising properties to encode quantum information into their degenerate groundspace. The different ground states can be labeled through a topological property of the system, e.g., by the equivalency classes of the different non-contractible loops on a torus. The quantum information encoded into a ground state can be recovered by performing error correction, even after a long time provided only local coherent errors are introduced by the environment. In this sense, topological systems are inherently robust to decoherence.

One of the first proposals for a topological quantum code is the \emph{toric code} by Kitaev \cite{Kitaev03}. This is a two-dimensional system with periodic boundary conditions, i.e., with a toroidal geometry, where physical spin-1/2 particles or \emph{qubits} live on edges of a 2D square lattice. This model has a four-fold degenerate groundspace, the groundspace encodes two logical qubits. Any local operator acts trivially within the groundspace whereas operators acting on a large number of qubits residing on a non-contractible loop going around the torus act non-trivially. An experimentally more feasible version of the toric code is the \emph{surface code} \cite{BK98,FMM+12}, which is a two-dimensional system with physical qubits still placed on edges of a lattice, but the boundaries are now open. Several experimental groups currently pursue the physical realization of surface codes~\cite{BKM+14,CMS+15} with the goal to use them as building blocks in a quantum computer.

The toric code belongs to a more general class of topological systems known as \emph{quantum doubles}, introduced by Kitaev~\cite{Kitaev03}. These are spin systems on a 2D lattice, whose excitations are point-like and they correspond to (non-abelian) anyons. Excitations of the quantum double of group $G$ correspond to the anyons described by the mathematical construction~\cite{Drinfeld86} known as the Drinfeld double $\mathcal{D}(G)$. For instance, the toric code is the quantum double $\mathcal{D}(\mathbb{Z}_2)$ based on the group $\mathbb{Z}_2$.

The excitations of a topological quantum field theories are indistinguishable quasi-particles called anyons: \emph{abelian} if taking anyons around each other modifies their wave function by only a phase, and \emph{non-abelian} if taking certain anyons around one another applies a nontrivial unitary operation to their wave function. In topological quantum field theories (TQFT), anyons carry a (nontrivial) charge or flux and are accordingly grouped into \emph{chargeons}, \emph{fluxons} and \emph{dyons} when they carry both a (nontrivial) charge and a (nontrivial) flux.

Quantum double models were introduced by Kitaev as a lattice realization of topological quantum field theories~\cite{Kitaev03}. Those models are defined by a Hamiltonian whose groundspace is spanned by vacuum states, i.e., states with no flux nor charge present. More precisely, the Hamiltonian imposes an energy penalty equal to the number of nontrivial charges or fluxes present.

Anyons are point-like excitations that appear on a site of the lattice. They are labeled by irreducible representations (irreps) of the Drinfeld double $\mathcal{D}(G)$. However, the spatial scale inherent to the lattice breaks the purely topological properties of the model and introduces \emph{local degrees of freedom}. In particular, anyon types are not in one-to-one correspondence with energy eigenspaces. Indeed, two anyons of the same type can have different energies depending on the local degrees of freedom. This peculiar feature is already present in Kitaev's original Hamiltonian but is more explicit in the family of Hamiltonians we introduce in this paper. Those novel Hamiltonians generalize Kitaev's original proposal since they have additional local terms which allow to distinguish among the different nontrivial fluxes and charges.    

In this paper, we introduce a family of Hamiltonians that assigns different energies to the different nontrivial fluxes and charges of non-abelian quantum doubles $\mathcal{D}(G)$. In these \emph{refined Hamiltonian}, each term only acts on four neighboring higher-dimensional spins (a.k.a qudits in the quantum information jargon). Moreover, each 4-local terms commute pairwise, resulting in a Hamiltonian which can be solved explicitly. Our construction is qualitatively different than the 6-local terms introduced in~\cite{BM08} since our family of Hamiltonian maintain the feature that Hamiltonian term are either related to the charges or to the fluxes. We then show how the 4-local charge and flux projectors assign different energies to excitations by partitioning the Hilbert space of excitations according to charge and flux labels related to the representation theory of the group $G$. Our construction emphasizes a feature already present in Kitaev's original proposal that anyons are not in one-to-one correspondence with energy eigenspaces of the Hamiltonian due to the presence of local degrees of freedom.

Throughout the paper, we illustrate the notions we introduced by analyzing the quantum double for the smallest non-abelian group $S_3$, the symmetry group of order 3, whose quantum double structure was explored in~\cite{BAC09,BSW11}. We explicitly write down the 4-local refined Hamiltonian for this theory, see Eq.~\eqref{eq:DS3_Hamiltonian}.

The paper is organized as follows. 
First, in Sec.~\ref{sec:quantum_double} we review the most important properties of non-abelian anyons, and introduce the quantum double construction. 
Second, in Sec.~\ref{sec:tunable_Hamiltonian} we introduce the general charge and flux projectors and construct the 4-local refined Hamiltonian, see Theorem~\ref{thm:main-thm}. 
We analyze how these projectors partition the Hilbert space of each site in Sec.~\ref{sec:Hilbert-splitting} and introduce a diagrammatic representation to visualize this partitioning, see Fig.~\ref{fig:anyon_partitioning}. 
This diagrammatic representation reveals that anyons are not in one-to-one correspondence with energy eigenspaces of the Hamiltonian due to the presence of local degrees of freedom. We explore how those local degrees of freedom arise out of the spatial scale introduced by the lattice in Sec.~\ref{sec:local-dof}.
Finally, we conclude our findings and point out future directions in Sec.~\ref{sec:conclusions}.

\section{The Drinfeld double construction and the quantum double models}
\label{sec:quantum_double}

The quantum double construction realizes topological lattice spin models whose anyonic excitations are described mathematically by the Drinfeld double of a group. To better appreciate the quantum double construction, we first review the properties and mathematical formalism of non-abelian anyons in general. First, in Sec.~\ref{subsec:aharonov-bohm} we give an overview of the anyon labels and the most important braiding properties. This pedagogical exposition is largely inspired from John Preskill's lecture notes~\cite{Preskill98} and the reader is encouraged to consult those notes for more details. Then we introduce the quantum double construction on a lattice in Sec.~\ref{subsec:quantum_double}. 

\subsection{Non-Abelian Aharonov-Bohm effect}
\label{subsec:aharonov-bohm}

Anyons can be understood by analogy to the Aharonov-Bohm effect: taking a charge $q$ around a flux tube with flux $\Phi$ results in the wave function acquiring a phase $\exp (i q \Phi)$.
\begin{equation}
 \ket{\psi} \to \exp (i q \Phi) \ket{\psi}
\end{equation}

Non-Abelian anyons can be qualitatively understood by generalizing the Aharonov-Bohm effect to fluxes whose possible values correspond to the elements $g$ of a group $G$ and the charge possible values are the irreducible representations (irreps) $\Gamma$ of $G$. In other words, the Hilbert space of each quasiparticle is spanned either by the flux orthonormal basis
\begin{equation}
 \cH=\mbox{span} \{\ket{g}\}_{g\in G}.
\end{equation}
or in a conjugate charge orthonormal basis 
\begin{equation}
 \cH=\mbox{span} \{\ket{\Gamma,i}\}_{\mathrm{irrep }\Gamma,i=1\dots |\Gamma|}
\end{equation}
in which we chose an (arbitrary) orthonormal basis $\{\ket{\Gamma,i}_{i=1\dots |\Gamma|}\}$ for every module of each irrep $\Gamma$.

\subsubsection{Labeling fluxons}

To identify a fluxon, we can check how the basis transforms when a charge $\Gamma$ is transported around the fluxon
\begin{equation}\label{eq:interferometry}
 \ket{\Gamma,j} \to \sum_{i=1}^{|\Gamma|} D^{\Gamma}_{ij}(a) \ket{\Gamma,i}
\end{equation}
Since the matrix elements $D^{\Gamma}_{ij}(a)$ can in principle be measured by interferometry~\cite{Bonderson07}, performing this for every charge type $\ket{\Gamma,j}$ will reveal the flux $a\in G$.

However, labeling fluxons by group elements is not gauge-invariant since another observer could choose another orthonormal basis for the module of the irrep $\Gamma$. In fact, the correct gauge-invariant quantity to label fluxons is the \emph{conjugacy class}:
\begin{defn}[Conjugacy class]
\be
C_a = \{ g a g^{-1} | g \in G \} .
\ee
\end{defn}
Indeed, two observers will agree on the conjugacy class of a fluxon even if they probably would disagree on the representative group element within the conjugacy class.

\subsubsection{Braiding of fluxons}

We now want to understand what happens when braiding fluxons. Let's consider two fluxons side by side. The left fluxon has flux $a$ while the right fluxon has flux $b$ (locally, flux types are well defined). Let's now counterclockwise exchange the fluxons, resulting in an operator $R_{ab}$. One can prove that the resulting effect is 
\begin{equation}
 R_{ab} \: : \ket{a,b} \mapsto \ket{aba^{-1},a}
\end{equation}
i.e., the right flux has been conjugated by the left flux. See Fig.~\ref{fig:braiding} for a pictorial representation.

\begin{figure}
\begin{centering}
\includegraphics[width=0.3\textwidth]{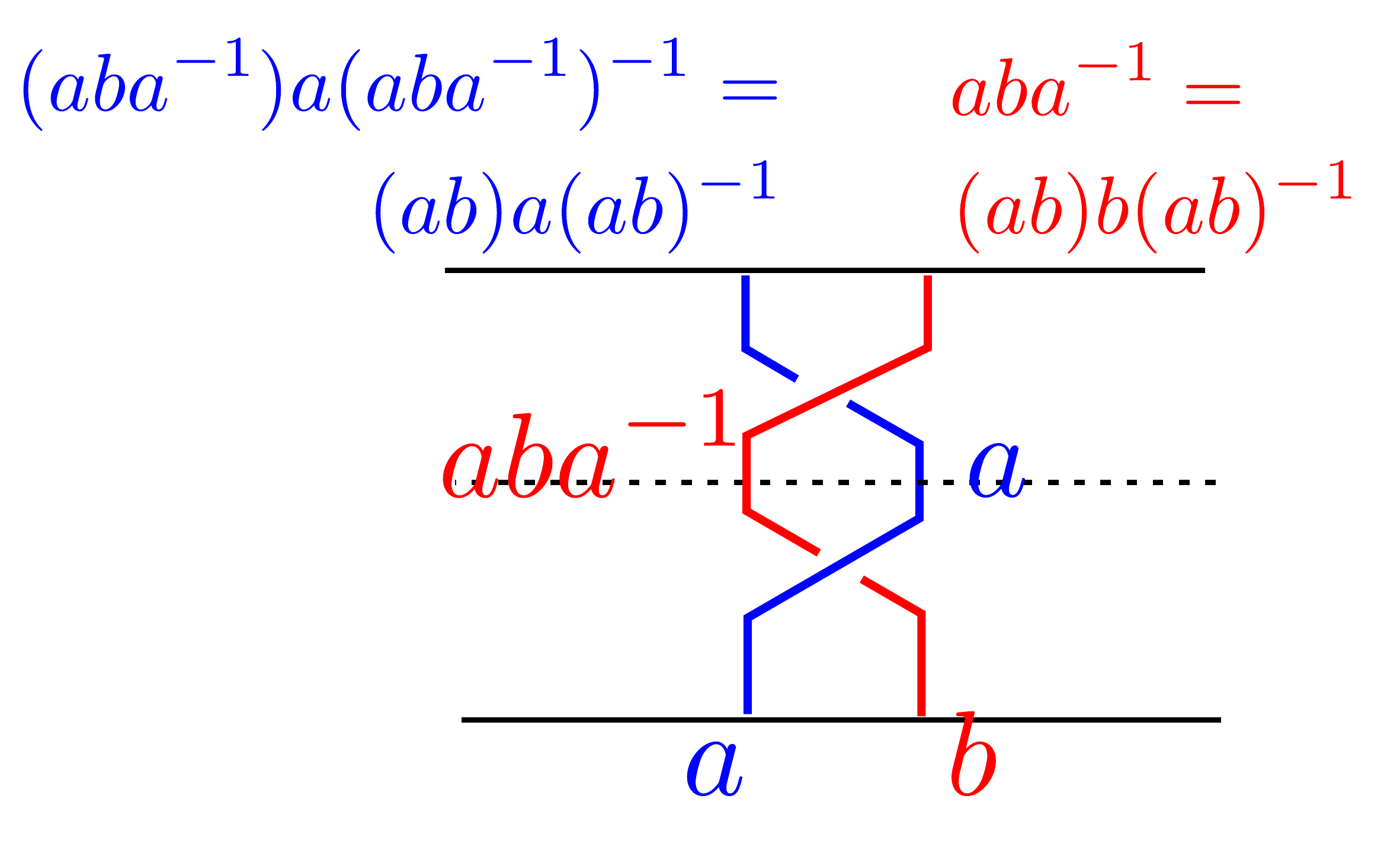}
\caption{(Color online) Braiding of two anyons, $a$ and $b$: applying a counterclockwise exchange of the particles, resulting in conjugacy of the original wave function.}
\label{fig:braiding}
\end{centering}
\end{figure}

Note that two successive counterclockwise exchange is equivalent to having the rightmost flux going around the leftmost flux counterclockwise, see Fig.~\ref{fig:braiding}. The net result of that operation is 
\begin{equation}
 R^2_{ab} \: : \ket{a,b} \mapsto \ket{(ab)a(ab)^{-1},(ab)b(ab)^{-1}}
\end{equation}
which is coherent with the claim that the conjugacy class of a fluxon is gauge-invariant but the representative is ambiguous since it can change by an arbitrarily far away fluxon moving around it.

\subsubsection{Dyon: anyon with nontrivial flux and nontrivial charge}

While we have discussed how to label a chargeon (by an irrep) and a fluxon (by a conjugacy class), we have yet to discuss anyons that exhibit both a nontrivial charge and a nontrivial flux. Such an anyon is called a \emph{dyon}. Suppose we wanted to measure the charge of a dyon. We could set up an interferometric experiment. We could place the dyon behind the slits in a double slit experiment and measure the interferometry pattern for any incoming test fluxon. However, since the dyon also carries flux, subtleties arise. Indeed, the passage of the test fluxon either to the left or the right of the dyon will modify the flux of the dyon. Thus, interference will only occur if the flux $a$ of the dyon commutes with the flux $b$ of the test fluxon, i.e., if $ab=ba$. In other words, the charge $\Gamma$ of the dyon can be determined only if the probe fluxon has a flux among the elements $b$ commuting with $a$, i.e., within the \emph{normalizer} of $a$ 
\begin{defn}[Normalizer]
\be
\mathcal{N}_a = \{ b \in G | ab = ba \} .
\ee
\end{defn}
Note that a normalizer is always a subgroup of the group $G$. We thus conclude that the charge $\Gamma$ of a dyon carrying flux $a$ is not an irrep of the full $G$, but rather an irrep of the normalizer $\mathcal{N}_a$. 

The mathematical structure corresponding to an anyon model is the Drinfeld double of a group which is a quasi-triangular Hopf algebra. Anyon types are in one-to-one correspondence with the irreps of that operator algebra. 
Working out the irreps of the Drinfeld double only requires knowledge of the representation theory of the underlying group, since a key mathematical result is that irreps of a Drinfeld double are labeled by i) a conjugacy class and ii) an irrep of the normalizer of any element of the conjugacy class (which are all isomorphic). 

\subsubsection{Quantum dimension of an anyon}

In a Drinfeld double, the quantum dimension $d_a$ associated to every anyon type $a$ is the dimension of the vector subspace associated to that anyon. It is thus an integer. Given an anyon type $(C_g,\,\Gamma )$, its quantum dimension is 
\begin{equation}
 d_{(C_g,\,\Gamma )} = |C_g| |\Gamma|. \label{eq:Qdim-anyon}
\end{equation}
Moreover, another quantity of interest is the total quantum dimension $\mathcal{D}$ of the model, which is related to the quantum dimension of every anyon type by
\begin{equation}
 \mathcal{D}^2= \sum_{\mathrm{anyons}\:k} d_k^2. \label{eq:def-totalQdim}
\end{equation}
In the case of a quantum double, the total quantum dimension is related to the cardinality of the group
\begin{equation}
 \mathcal{D}^2= |G|^2. \label{eq:totalQdim-Qdouble}
\end{equation}
This result might appear as mysterious: we will give an interpretation of this result in Sec.~\ref{sec:diagram}.

\subsubsection{Example of $\mathcal{D}(S_{3})$}
\label{subsec:D(S_3)}

As a more elaborate example of the above quantum double structure, let's look at the quantum double of the smallest non-abelian group, $\mathcal{D}(S_{3})$. The group $S_3$ is isomorphic to the symmetry transformations of an equilateral triangle (see Fig.~\ref{fig:triangle_symmetries}):
\begin{itemize}
	\item identity: $e$,
	\item rotations by $\pi/3$ and $2 \pi/3$: $y$, $y^2$,
	\item mirrorings to the three different axes: $x$, $xy$, $xy^2$ .
\end{itemize}

Because of the nature of these symmetries: $y^3 = e$ and $x^2 = (xy)^2 = (xy^2)^2 = e$. The non-abelianity of $\mathcal{S}_3$ is summed up by the commutation relation $xy=y^2 x$.

\begin{figure}
\begin{centering}
\includegraphics[width=0.2\textwidth]{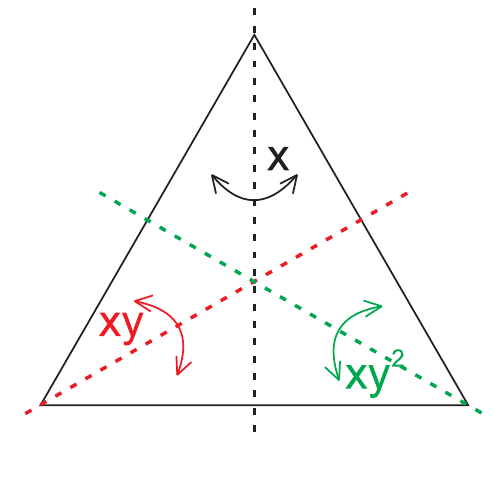}\hspace{0.5cm}
\includegraphics[width=0.2\textwidth]{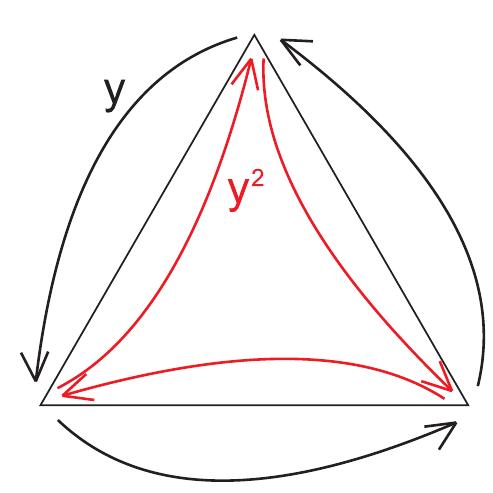}
\caption{(Color online) Symmetries of an equilateral triangle, or elements of the group $S_3$.}
\label{fig:triangle_symmetries}
\end{centering}
\end{figure}

The anyons of the Drinfeld double of $S_3$ are labeled by the conjugacy classes of $S_3$ and the irreducible representations of normalizers of conjugacy classes. There are three conjugacy classes of $S_3$:
\bq
C_e &=& \{ e \}, \\
C_y &=& \{ y, y^2 \}, \\
C_x &=& \{ x, xy, xy^2 \},
\eq
and the corresponding normalizers are
\bq
\cN_e &=& S_3 ,\\
\cN_y = \cN_{y^2} &=& \{ e, y, y^2 \} \cong \mathbb{Z}_3 ,\\
\cN_x &=& \{ e, x \} \cong \cN_{xy} \cong \cN_{xy^2} \cong \mathbb{Z}_2 .
\eq

We would like to point out here that while the normalizers $\cN_y$ and $\cN_{y^2}$ are the same, independent of the labeling, $\cN_x$, $\cN_{xy}$ and $\cN_{xy^2}$ are distinct, and only isomorphic to each other. 

The irreducible representations of all these normalizers are listed in Tables~\ref{tab:irreps_of_S3}-\ref{tab:irreps_of_Z3_Z2}. There and in the remainder of the paper $\omega = \exp(2\pi i/3)$ and $\bar{\omega} = \exp(4\pi i/3)$ are the third complex roots of unity.

\begin{table}
	\centering
		\begin{tabular}{|c||cccccc|}
			\hline
			$S_3$ & $e$ & $y$ & $y^2$ & $x$ & $xy$ & $xy^2$ \\
			\hline
			\hline
			$\Gamma_1^{S_3}$ & 1 & 1 & 1 & 1 & 1 & 1 \\
			\hline
			$\Gamma_{-1}^{S_3}$ & 1 & 1 & 1 & -1 & -1 & -1 \\
			\hline
			$\Gamma_{2}^{S_3}$ & $\begin{pmatrix} 1 & 0 \\ 0 & 1 \end{pmatrix}$ 
			& $\begin{pmatrix} \bar{\omega} & 0 \\ 0 & \omega \end{pmatrix}$ 
			& $\begin{pmatrix} \omega & 0 \\ 0 & \bar{\omega} \end{pmatrix}$ 
			& $\begin{pmatrix} 0 & 1 \\ 1 & 0\end{pmatrix}$ 
			& $\begin{pmatrix} 0 & \omega \\ \bar{\omega} & 0\end{pmatrix}$ 
			& $\begin{pmatrix} 0 & \bar{\omega} \\ \omega & 0\end{pmatrix}$ \\
			\hline
		\end{tabular}
		\caption{Irreducible representations of $S_3$, i.e. the possible charge labels with flux $C_e$.}
		\label{tab:irreps_of_S3}
\end{table}

\begin{table}
	\centering
	\begin{tabular}{|c||ccc|}
			\hline
			$\mathbb{Z}_3$ & $e$ & $y$ & $y^2$ \\
			\hline
			\hline
			$\Gamma_1^{\mathbb{Z}_3}$ & 1 & 1 & 1 \\
			\hline
			$\Gamma_{\omega}^{\mathbb{Z}_3}$ & 1 & $\omega$ & $\bar{\omega}$ \\
			\hline
			$\Gamma_{\bar{\omega}}^{\mathbb{Z}_3}$ & 1 & $\bar{\omega}$ & $\omega$ \\
			\hline
		\end{tabular}
		\hspace{1cm}
		\begin{tabular}{|c||cc|}
			\hline
			$\mathbb{Z}_2$ & $e$ & $x$ \\
			\hline
			\hline
			$\Gamma_1^{\mathbb{Z}_2}$ & 1 & 1 \\
			\hline
			$\Gamma_{-1}^{\mathbb{Z}_2}$ & 1 & -1 \\
			\hline
		\end{tabular}
		\caption{Irreducible representations of (a) $\mathbb{Z}_3$ and (b) $\mathbb{Z}_2$, i.e. the possible charge labels with flux $C_y$ and $C_x$.}
		\label{tab:irreps_of_Z3_Z2}
\end{table}

In summary, this model has 8 anyons, these are listed in Table~\ref{tab:D(S3)_anyons}. Anyon A is the vacuum since it has both trivial charge and flux. Anyons B and C are chargeons, they correspond respectively to the signed and two-dimensional irreps of $S_3$. Anyons $D$ and $F$ are fluxons since they correspond to the trivial irrep of their respective normalizers. Other anyons are dyons.

\begin{table}
	\centering
		\begin{tabular}{|c||ccccc|}
			\hline
			Label & $C_g$ & $\mathcal{N}_g$ & Irrep. & Q.dim. & Type \\
			\hline
			\hline
			A & $C_e$ & $S_3$ & $\Gamma^{S_3}_1$ & 1 & vacuum \\
			\hline
			B & $C_e$ & $S_3$ & $\Gamma^{S_3}_{-1}$ & 1 & chargeon \\
			\hline
			C & $C_e$ & $S_3$ & $\Gamma^{S_3}_2$ & 2 & chargeon \\
			\hline
			\hline
			D & $C_x$ & $\mathbb{Z}_2$ & $\Gamma^{\mathbb{Z}_2}_1$ & 3 & fluxon \\
			\hline
			E & $C_x$ & $\mathbb{Z}_2$ & $\Gamma^{\mathbb{Z}_2}_{-1}$ & 3 & dyon \\
			\hline
			\hline
			F & $C_y$ & $\mathbb{Z}_3$ & $\Gamma^{\mathbb{Z}_3}_1$ & 2 & fluxon \\
			\hline
			G & $C_y$ & $\mathbb{Z}_3$ & $\Gamma^{\mathbb{Z}_3}_{\omega}$ & 2 & dyon \\
			\hline
			H & $C_y$ & $\mathbb{Z}_3$ & $\Gamma^{\mathbb{Z}_3}_{\bar{\omega}}$ & 2 & dyon \\
			\hline
		\end{tabular}
		\caption{Anyons of $\mathcal{D}(S_3)$ with their charge and flux labels, quantum dimensions and type.}
		\label{tab:D(S3)_anyons}
\end{table}

At this point, we have defined anyons and described their braiding and fusion properties using a toy model of non-abelian Aharonov-Bohm effect. We recovered, using a physics point of view, the key properties of the Drinfeld double of a group. In particular, we worked out in detail the anyon types of $\mathcal{D}(S_3)$. However, in this toy model, anyons are fundamental particles. We will now describe the quantum double construction by Kitaev in which those anyons appear effectively as point-like excitations on a spin lattice.

\subsection{Kitaev's quantum double on a lattice}
\label{subsec:quantum_double}

A way to realize the non-abelian Aharonov-Bohm effect on a lattice is Kitaev's quantum double construction \cite{Kitaev03}. In this construction, charges reside on vertices and fluxes are on plaquettes of the lattice, however, fluxes and charges are not independent. A generic flux-charge composite particle (dyon) lives on a site: a combination of a vertex and a plaquette shown in Fig.~\ref{fig:site_def}.

This excitation structure is realized by first, assigning a Hilbert space to each edge of the lattice, the state of each edge can take any group element $z \in G$, then, defining a Hamiltonian that describes the interactions in this model. To introduce the Hamiltonian, let us define the following operators:
\bq
L^+_g \ket{z} &=& \ket{gz} ,\\
L^-_g \ket{z} &=& \ket{zg^{-1}} ,\\
T^+_h \ket{z} &=& \delta_{h,z} \ket{z} ,\\
T^-_h \ket{z} &=& \delta_{h^{-1},z} \ket{z} ,
\eq
where $L^+_g$ and $L^-_g$ are the matrices representing left- and right-multiplication operators, $T^+_h$ and $T^-_h$ are diagonal operators in the flux basis.

Then, we need to assign an orientation to the edges of the lattice. We use the convention shown in Fig.~\ref{fig:site_def} for a site, i.e., the union of a vertex and a plaquette.

 \begin{figure}
\begin{centering}
\includegraphics[width=0.2\textwidth]{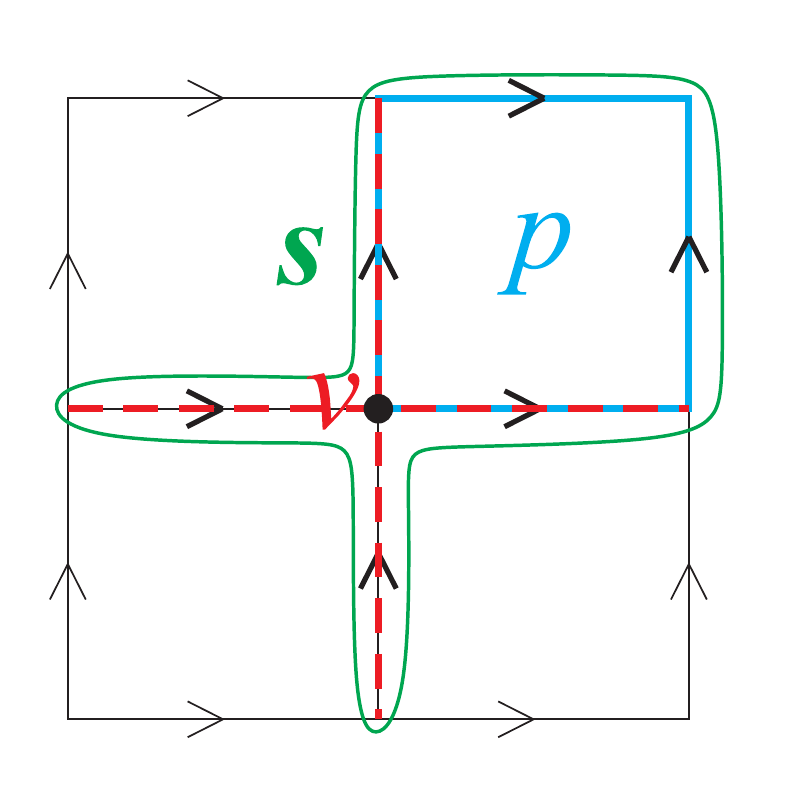}
\includegraphics[width=0.2\textwidth]{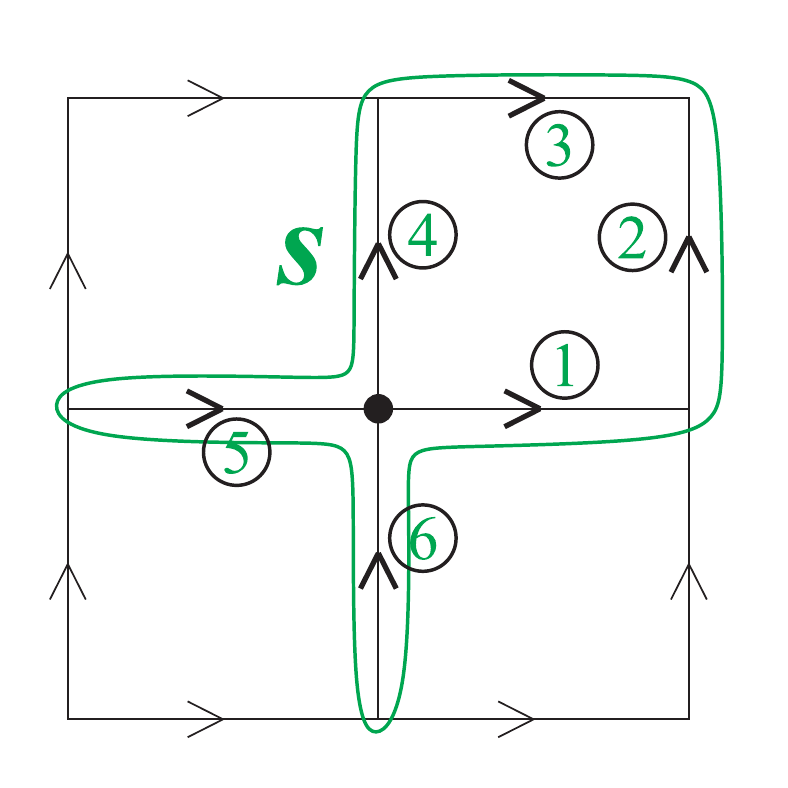}
\caption{(Color online) Our choice of orientation on the lattice, with (a) how a vertex $v$ and plaquette $p$ form a site $s$, and (b) the edge numbering we used to define the vertex and plaquette operators $\mathcal{A}^v_g$ and $B^p_h$ in Eqs.~\eqref{eq:plaq_op_def}-\eqref{eq:vertex_op_def}.}
\label{fig:site_def}
\end{centering}
\end{figure}

We now introduce two families of operators, following closely the original definition of~\cite{Kitaev03}.
\begin{defn}[Plaquette operators]
 For any element $h\in G$, we define an operator acting on the 4 spins around an oriented plaquette $p$
 \begin{equation} \label{eq:plaq_op_def}
  B^p_{h} = \sum_{h_1 h_2 h_3^{-1} h_4^{-1} = h} T^{+,1}_{h_1} \otimes T^{+,2}_{h_2} \otimes T^{-,3}_{h_3} \otimes T^{-,4}_{h_4}
 \end{equation}
 where the use of $T_h^{\pm}$ depends on the respective orientations of the plaquette and the edges. See Fig.~\ref{fig:site_def} for our orientation convention and the labeling of the spins.
 
 \end{defn}

 \begin{defn}[Vertex operators]
 For any element $g\in G$, we define a vertex operator, originally called star operators in~\cite{Kitaev03}, acting on the 4 spins around a vertex $v$
 \begin{equation} \label{eq:vertex_op_def}
  \mathcal{A}_g^v =  L^{+,1}_g \otimes L^{+,4}_g \otimes L^{-,5}_g \otimes L^{-,6}_g
 \end{equation}
where $L^+_g$ appears for outgoing edges and $L^-_g$ appears for incoming edges. See Fig.~\ref{fig:site_def} for our orientation convention and the labeling of the spins. 
\end{defn}
 
How these operators act on a vertex and on a plaquette is illustrated in Fig.~\ref{fig:A_B_acting_on_sites}. In order for individual $B^p_{h}$ to be properly defined even for a non-abelian group, we need to specify a starting vertex on the plaquette, then specify an orientation. Henceforth, we mark the starting vertex by a black dot in Figures~\ref{fig:site_def}-\ref{fig:A_B_acting_on_sites} and systematically orient the plaquettes in a counterclockwise manner. Whenever the orientation of an edge is opposite to the orientation of the plaquette, a plaquette operator $B_h$ acts on it with $T^-_h$, otherwise it acts with $T^+_h$. Similarly for the vertex operators: when the orientation of an edge points outwards from the vertex, $\mathcal{A}_g^v$ acts with $L^+_g$, otherwise with $L^-_g$ on that edge.

The projector unto the trivial flux at plaquette $p$ is simply the plaquette operator for the trivial element $B^p_{e}$. The projector unto trivial charge $A^v_{1}$ on vertex $v$ is defined as
\be
A^v_{1} = \sum_{g \in G} \mathcal{A}_g^v = \sum_{g \in G} L^{+,1}_g \otimes L^{+,4}_g \otimes L^{-,5}_g \otimes L^{-,6}_g,
\ee
 where the use of $L^+_g$ vs. $L^-_g$ again depends on the orientation of the edge with respect to the vertex. It is less trivial to see why this operator projects to the trivial charge, i.e., corresponds to the trivial representation. One explanation is that for any $g \in G$, we have $A^v_1 \mathcal{A}^v_g = A^v_1$. Thus, the image of $A^v_1$ is invariant under the action of any $\mathcal{A}^v_g$, which is characteristic of the trivial representation.

\begin{figure}
\begin{centering}
\includegraphics[width=0.4\textwidth]{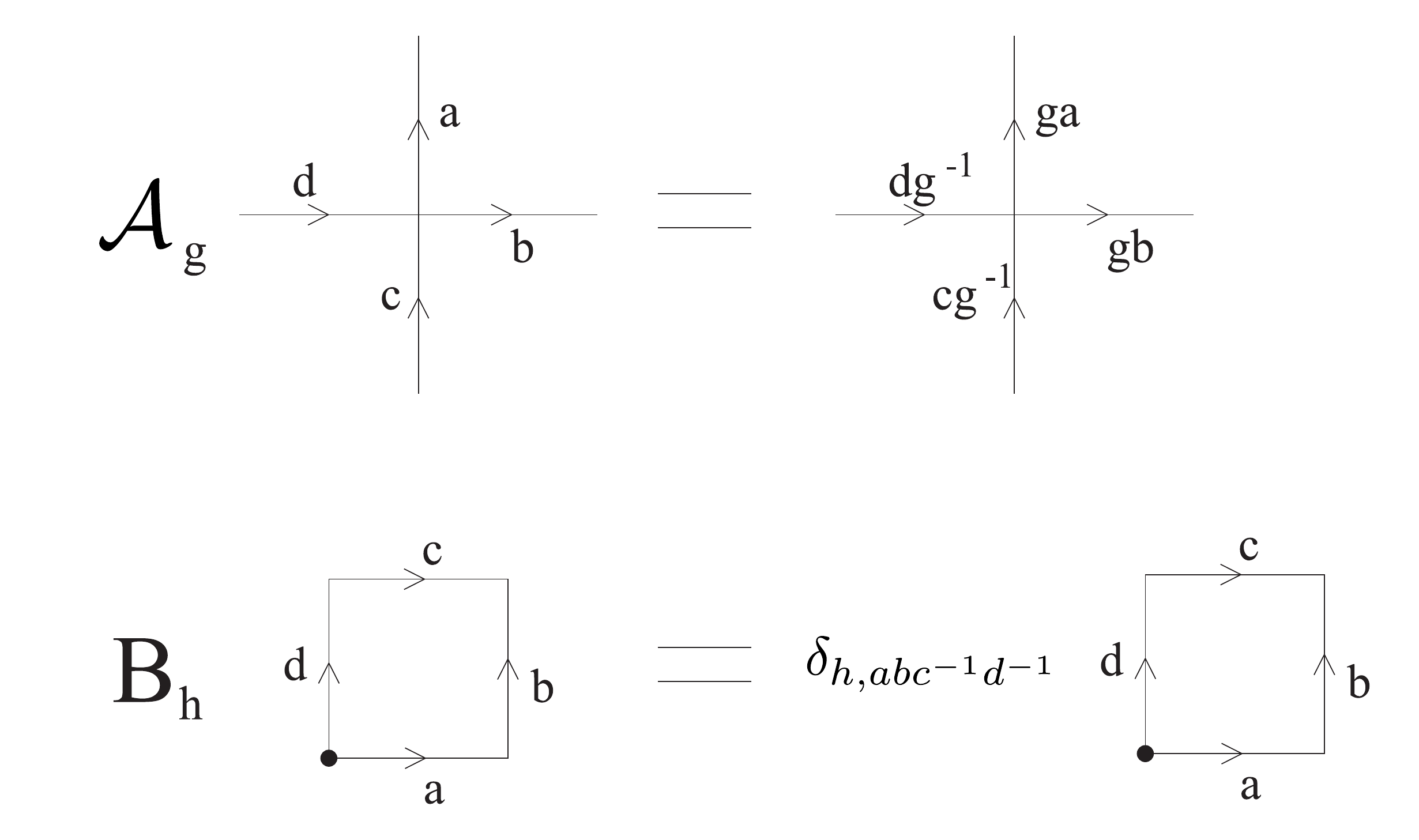}
\caption{The effect of the individual projector terms ($\mathcal{A}_g$ and $B_{h}$) on a vertex and on a plaquette, respectively.}
\label{fig:A_B_acting_on_sites}
\end{centering}
\end{figure}

Given vertex and plaquette operators, Kitaev introduced the following Hamiltonian in~\cite{Kitaev03}. 

\begin{defn}[Kitaev Hamiltonian]
The Kitaev Hamiltonian of a quantum double $\mathcal{D}(G)$ is
\be \label{eq:vacuum_Hamiltonian}
H = - \sum_v A^v_{1} - \sum_p B^p_{e} ,
\ee
\end{defn}

Please note that Hamiltonian \eqref{eq:vacuum_Hamiltonian} assigns an extensive energy of $-2$ for every site in the vacuum (ground state). Any vertex which does not carry the trivial charge receives an energy penalty. Similarly, any plaquette which does not exhibit a trivial flux receives an energy penalty.

\subsubsection*{Example: Toric code}

The simplest example of the above quantum double construction is the toric code \cite{Kitaev03}. This is the quantum double of $\mathbb{Z}_2$, thus the possible group elements on an edge can be: $\{ 0, 1 \}$, and all additions are understood modulo $2$: $0 \oplus 0 = 0$, $0 \oplus 1 = 1$, $1 \oplus 1 = 0$. The corresponding spin states $\ket{0}$ and $\ket{1}$ are the usual computational basis for qubits.

In $\mathbb{Z}_2$, the left- and right-multiplication operators are the same: $L_0^+ = L_0^- = \id$ and $L_1^+ = L_1^- = X$, where $X$ is the Pauli $X$ operator. The diagonal operators are: $T_0^+ = T_0^- = (\id+Z)/2$ and $T_1^+ = T_1^- = (\id-Z)/2$, with $Z$ being the Pauli $Z$ operator.

The operators projecting unto trivial flux and trivial charge are (omitting the tensor product sign for simplicity, please refer to Fig.~\ref{fig:site_def} for the labeling convention): $A^v = \id^{1} \id^{4}  \id^{5}  \id^{6} + X^{1}  X^{4}  X^{5}  X^{6}$ and $B^p = \id^{1}  \id^{2}  \id^{3}  \id^{4} + Z^{1}  Z^{2}  Z^{3}  Z^{4}$, thus the Hamiltonian is
\bq
H = - \sum_v (\id^{1}  \id^{4}  \id^{5}  \id^{6} + X^{1}  X^{4}  X^{5}  X^{6}) \nonumber \\ 
- \sum_p (\id^{1}  \id^{2}  \id^{3}  \id^{4} + Z^{1}  Z^{2}  Z^{3}  Z^{4}) ,
\eq
or in its widely known form, after redefining the ground state energy:
\be \label{eq:toric_code_H}
H = - \sum_v X^{1}  X^{4}  X^{5}  X^{6} - \sum_p Z^{1}  Z^{2}  Z^{3}  Z^{4} .
\ee

Similar to the general quantum double Hamiltonian of Eq.~\eqref{eq:vacuum_Hamiltonian}, this Hamiltonian assign an extensive energy of $-2$ to all sites in the vacuum state, and an energy penalty for vertices with a non-trivial charge, and for plaquettes with a non-trivial flux. Having a non-trivial charge or flux at a certain vertex/plaquette is frequently referred to as "violating" that vertex/plaquette term in the Hamiltonian; the eigenvalue of each 4-body term is either $+1$ (no charge/flux) or $-1$ (charge/flux excitation). A violated vertex term corresponds to a charge excitation ($e$), while a violated plaquette term means a flux excitation is living on that plaquette ($m$). The four possible states of a site are thus:
\bq \nonumber
&& \id \textrm{  (vacuum)} \\
\nonumber
&& e \textrm{  (charge)} \\
\nonumber
&& m \textrm{  (flux)} \\
\nonumber
&& \epsilon = e \otimes m \textrm{  (dyon)}
\eq

Further analyzing the Hamiltonian of Eq.~\eqref{eq:toric_code_H}, we can see two main features of the model. First, the charges and fluxes have decoupled from each other, which is typical of abelian quantum doubles since any anyon type is the juxtaposition of a charge and a flux. Indeed, the only dyon is $\epsilon = e \otimes m$, which is the simple combination of the non-trivial charge ($e$) and the non-trivial flux ($m$), and has no additional emergent properties.

Second, there is only one kind of excitation of either type (1 electric charge ($e$) and 1 magnetic flux ($m$)) in this model. Therefore, the Hamiltonian, which contains only two projectors can distinguish the four types of anyons: vacuum $\id$, electric chargeon $e$, magnetic fluxon $m$ and the dyon $\epsilon=e\otimes m$. We will refer to this as the Kitaev Hamiltonian having 1-to-1 correspondence between energy eigenstates and anyon type, in the case of the toric code. 

We will now generalize the Kitaev Hamiltonian for non-abelian models by introducing local terms which project unto the different possible charges and the different possible fluxes in Sec.~\ref{sec:tunable_Hamiltonian}. We will argue that this is a natural generalization of Kitaev's Hamiltonian. However, our generalization will highlight that anyon type and energy eigenspaces are \emph{not} in one-to-one correspondence for \emph{non-abelian} quantum doubles in Sec.~\ref{sec:Hilbert-splitting}. This peculiar feature was already present in Kitaev's original construction. This discrepancy between energy eigenspaces and anyons stems from the presence of \emph{local degrees of freedom} that are not topological and arise from the lattice, in the case of non-abelian models. Those will be explored in Sec.~\ref{sec:local-dof}.

\section{Refined quantum double Hamiltonian for arbitrary group}
\label{sec:tunable_Hamiltonian}

We have seen in the previous section that the Kitaev Hamiltonian given by Eq. ~\eqref{eq:vacuum_Hamiltonian} assigns an energy penalty to any nontrivial charge and flux. However, it does not distinguish among two distinct nontrivial charges or fluxes. It is then natural to wonder whether one can enrich the model by introducing new local terms which will introduce such a distinction? And if yes, how will that change the excitation structure of the theory?

In this section, we introduce in Sec.~\ref{sec:tunable-construction} a Hamiltonian that splits up the energies of different excitations for any quantum double, and then, in Sec.~\ref{subsec:D(S_3)_Hamiltonian}, we work out explicitly the corresponding Hamiltonian for the quantum double of $\mathcal{D}(S_3)$.

\subsection{Refined quantum double construction}\label{sec:tunable-construction}

Our aim in this section is to introduce projectors unto different nontrivial charges and fluxes. In Sec.~\ref{subsec:quantum_double} we have already given the form of the trivial flux projector $B^p_e$ and trivial charge projector $A^v_1$. Even though these vacuum projectors are independent of one another, and are both 4-body operators, it is not trivial that one could write independent 4-body charge- and flux-projectors. This is because, unlike in the case of abelian quantum doubles, the charge and flux of a site are tied to one another when considering dyons, i.e. the charge is defined as an irreducible representation of the normalizer of the flux conjugacy class.

This section is organized as follows. We will first comment on the reasons we insist on defining a Hamiltonian whose terms are four-local in Sec.~\ref{sec:locality}. We then outline our construction by recalling the definition of flux projectors and introducing charge projectors in Sec.~\ref{sec:flux-charge-projectors}. This allows us to  define our family of refined Hamiltonians in Sec.~\ref{sec:definition-Hamiltonian}. Namely, Theorem~\ref{thm:main-thm} is the novel family of Hamiltonians introduced by our work. In the following sections, we sketch the proof of Theorem~\ref{thm:main-thm} which relies on proving that the charge projectors are indeed an orthonormal family of projectors in Sec.~\ref{sec:math-proof} and then proving that they commute with the flux projectors in Sec.~\ref{subsubsec:Commutation}. Formal mathematical proofs are given in the appendix.

\subsubsection{Locality of the Hamiltonian}\label{sec:locality}

A simple route to assign different masses to each anyon type would be to introduce a 6-local Hamiltonian. Indeed, each anyon lives on a site comprised of 6 spins. We can thus achieve an energy spectrum in one-to-one correspondence with anyon types by introducing 6-local projectors acting on sites, projecting unto the different anyon species defined by the combination of a flux and a charge label, $P^s_{(C_h, \Gamma^{\cN_h})}$ where $\Gamma^{\cN_h}$ labels irreps of the normalizer of each elements of $C_h$ (which are isomorphic). 

More precisely, we can define a massive 6-local Hamiltonian 
\be \label{eq:massiveH_6body}
H = \sum_s \sum_{C_h} \sum_{\mathrm{irrep} \: \Gamma^{\cN_h}} \alpha_{(C_h,\Gamma^{\cN_h})}  P^s_{(C_h,\Gamma^{\cN_h})} ,
\ee
where the projectors have been defined in~\cite{BM08}
\begin{equation} \label{eq:6-body_projector}
 P^{s=(p,v)}_{(C_h, \Gamma^{\cN_h})} = \sum_{g\in C_h} \sum_{g'\in \cN_g} \frac{d_{\Gamma^{\cN_{g}}}}{|\cN_g|} \chi_{\Gamma^{\cN_{g}}}(g') \cA^v_{g'} B^p_g ,
\end{equation}
where  $d_{\Gamma^{\cN_{g}}}$ is the dimension of irrep $\Gamma^{\cN_{g}}$ (an irrep of the normalizer $\cN_{g}$) and $\chi_{\Gamma^{\cN_{g}}}(g')=\mbox{Tr}\left[\Gamma^{\cN_{g}}(g')\right]$ is the character of group element $g'$ in irrep $\Gamma^{\cN_{g}}$.

Thus, each coupling constants $\alpha_{(C_h, \Gamma^{\cN_h})}$ corresponds to the mass of an anyon type and they can be tuned independently. While this Hamiltonian offers the greatest flexibility for the energy spectrum, we will follow a different construction for three main reasons: 
\begin{enumerate}
 \item First, we aim to have the non-abelian massive Hamiltonian be as close in form to the original Kitaev construction as possible, and we can achieve this without making our Hamiltonian more non-local. 
  \item The second reason for 4-local terms in the Hamiltonian is that we would like our Hamiltonian to remain local since it appears to be physically more realistic. And even though it might be possible to further decrease the degree of locality to 3-local commuting terms for non-abelian models~\cite{AE11}, we have arguments that indicate that 2-local commuting Hamiltonians cannot be topological in 2D~\cite{BV05}. Indeed, the 4-local toric code Hamiltonian can be recovered effectively in the right parameter regime of a nearest-neighbor 2-local, yet frustrated, Hamiltonian on a honeycomb lattice. More generally, there is a procedure to turn a 4-local quantum double Hamiltonian for arbitrary group into a frustrated 2-local Hamiltonian thanks to a so-called 'gadget construction'~\cite{BFB+11}. 

  \item Third, writing 4-local terms will allow us to underline the discrepancy between anyon types and energy eigenspaces arising from the emergence of local degrees of freedom.
\end{enumerate}

\subsubsection{Flux and charge projectors}\label{sec:flux-charge-projectors}

The operators acting on a plaquette and projecting to a specific flux/specific group element have already been introduced in Eq.~\eqref{eq:plaq_op_def}. However, as pointed out earlier, a group element does not provide a gauge-invariant labeling of fluxons. Thus, we are lead to define a flux projector by considering a conjugacy class $C_h$
\begin{defn}[Flux projectors]
 The flux projector associated to a conjugacy class $C_h$ of a group $G$ is
 \bq \label{eq:flux_projector}
B_{C_h} &=& \sum_{h' \in C_h} B_{h'}.
\eq
\end{defn}

We now introduce a family of charge projectors which generalizes the projector unto the trivial irrep introduced by Kitaev in~\cite{Kitaev03}. These \emph{charge projectors} are cornerstones of our refined quantum double construction.

\begin{defn}[Charge projectors]
The charge projector associated to an irreducible representation $\Gamma$ of the group $G$ is
\be \label{eq:charge_projector}
A_\Gamma = \frac{d_\Gamma}{|G|} \sum_{g \in G} \chi_\Gamma (g) \mathcal{A}_g ,
\ee
where  $d_\Gamma$ is the dimension of irrep $\Gamma$ and $\chi_{\Gamma}(g)=\mbox{Tr}\left[\Gamma(g)\right]$ is the character of group element $g$ in irrep $\Gamma$.
\end{defn}

These charge projectors can be thought of as a special case of the 6-body projector introduced in \cite{BM08} given by Eq.~\eqref{eq:6-body_projector}, with $C_h = C_e$. Using only the set of projectors defined in \eqref{eq:flux_projector} and in \eqref{eq:charge_projector} will lead to a different partitioning of the Hilbert space than by using the 6-body projectors of \eqref{eq:6-body_projector}. This will be further explored in Secs.~\ref{sec:Hilbert-splitting}-\ref{sec:local-dof}.

We defer a sketch of the proof that those operators are indeed orthogonal projectors to Sec. \ref{sec:math-proof}. One can check that for abelian groups, our charge projectors reduce to those introduced in Refs.~\cite{BAP14,KLT16}. Our charge projectors are reminiscent of similar objects introduced in~\cite{BAC09,WBI+14} using the representations themselves rather than the characters in the specific case of $\mathcal{D}(S_3)$.

\subsubsection{Definition of the refined quantum double Hamiltonian}\label{sec:definition-Hamiltonian}

Having defined flux projectors by Eq. \eqref{eq:flux_projector} and charge projectors by Eq. \eqref{eq:charge_projector}, we are now in a position to introduce our novel family of commuting Hamiltonians which assign different mass to different anyons.
\begin{thm} \label{thm:main-thm}
 The following family of topological Hamiltonians have commuting projector 4-local terms
 
 \begin{equation} \label{eq:massiveH_4body}
  H = \sum_v \sum_{\textrm{irrep } \Gamma^G} \alpha_{\Gamma^G} A^v_{\Gamma^G} + \sum_p \sum_{C_g \subset G} \beta_{C_g} B^p_{C_g}
 \end{equation}

\end{thm}

This family of commuting Hamiltonians is a central contribution of the paper. They are a new family of topological spin Hamiltonians made out of commuting projectors, similar to well-known families of topological models such as the Levin-Wen string-net models~\cite{LW05} and the Turaev-Viro codes~\cite{KKR10}. Compared to the Kitaev original quantum double Hamiltonians, they present the new feature of having tunable coupling constants that allow to discriminate among non-trivial charges and fluxes while preserving the useful mathematical properties of quantum doubles. In particular, the coupling constants can be chosen so that the groundspace is identical to the Kitaev Hamiltonian. Note that, for simplicity, we assumed the coupling coefficients to be independent of the vertices and the plaquettes, although they need not be.

We will now prove in Sec.~\ref{sec:math-proof} that the operators defined by Eq. \eqref{eq:charge_projector} are indeed projectors and then in Sec.~\ref{subsubsec:Commutation} that the charge and the flux projectors are pairwise commuting.

\subsubsection{Orthonormality of the charge projectors}\label{sec:math-proof}

\begin{thm}[Orthogonality of charge projectors]\label{thm:charge-projector}
The operators defined by Eq. \eqref{eq:charge_projector} are orthonormal projectors
\begin{equation}
A_{\Gamma}A_{\Lambda}=\delta_{\Gamma\Lambda}A_{\Gamma}
\end{equation}
\end{thm} 

\begin{proof}
This is a non-trivial consequence of the Great Orthogonality Theorem (GOT), see Fact~\ref{fact:GOT}. To prove this theorem, we will first prove a basis-independent statement of the GOT (Lemma \ref{lem:GOT-swap}). The full proof is deferred to the appendix in Sec.~\ref{subsec:proof-orthonormality}.
\end{proof}

\begin{fact}[Great Orthogonality Theorem]\label{fact:GOT}

\begin{equation}
\sum_{g\in G}\left(\Gamma(g)\right)_{ij}\overline{\left(\Lambda(g)\right)_{i'j'}}=\frac{|G|}{d_{\Gamma}}\delta_{\Gamma\Lambda}\delta_{ii'}\delta_{jj'}
\end{equation}
where $\overline{a}$ is the complex conjugate of $a\in\mathbb{C}$. 
\end{fact}

The Great Orthogonality Theorem is a strong result in representation theory, usually stated at the level of matrix elements of two representations $\Gamma$ and $\Lambda$ of a group $G$~\cite{Serre12}. In the proof of Theorem~\ref{thm:charge-projector} we utilize the following basis-independent version of the Great Orthogonality Theorem. To our knowledge, this operator restatement of the GOT is novel and could prove to be a useful tool in operator theory.

\begin{lem}[Basis-independent GOT]\label{lem:GOT-swap}
\begin{equation}
\sum_{g\in G}\Gamma(g)\otimes\Lambda(g^{-1})=\frac{\left|G\right|}{d_{\Gamma}}\delta_{\Gamma\Lambda}S
\end{equation}
where $S$ is the swap operator, i.e., $S:\mathbb{C}^{d}\times\mathbb{C}^{d}\to\mathbb{C}^{d}\times\mathbb{C}^{d}$
is defined by $S\left(|i\rangle\otimes|j\rangle\right)=|j\rangle\otimes|i\rangle$.\end{lem}

\begin{proof}
 The proof is deferred to the appendix in Sec.~\ref{subsec:proof-lemma-GOT}.
\end{proof}

\subsubsection{Commutation of flux and charge projectors} \label{subsubsec:Commutation}

We now prove that the flux projectors defined by Eq.~\eqref{eq:flux_projector} and charge projectors defined by Eq.~\eqref{eq:charge_projector} are pairwise commuting. This commutation is key since it entails that the two families of projectors split the Hilbert space in a consistent way, and states can be labeled by their common eigenstates.

\begin{lem}[Flux permutation by vertex operators] \label{lem:flux-permutation} For a plaquette $p$ and vertex $v$ that form a site, $(p,v)=s$
 \begin{equation}
 B^{(p)}_g = \cA^{(v)}_{h^{-1}} B^{(p)}_{hgh^{-1}} \cA^{(v)}_h ; \label{eq:commutation-lemma} 
 \end{equation}

for a plaquette $p$ and vertex $v$ that are parts of different sites, $p \in s_1$, $v\in s_2$, $s_1 \neq s_2$
 \begin{equation}
 B^{(p)}_g = \cA^{(v)}_{h^{-1}} B^{(p)}_{g} \cA^{(v)}_h . \label{eq:commutation-lemma2} 
 \end{equation}

\end{lem}

\begin{proof}
 The proof is deferred to the appendix in Sec.~\ref{subsec:proof-lemma-flux}.
\end{proof}

Based on Lemma~\ref{lem:flux-permutation} we can prove that vertex operators commute with flux projectors (although they do not commute with plaquette operators in general).

\begin{thm}
 \begin{equation}
  [ B_{C_g},\mathcal{A}_{h} ] = 0 \label{eq:commutation}
 \end{equation}

\end{thm}

\begin{proof}
 Lemma~\ref{lem:flux-permutation} shows that the vertex operators $\cA_h$ map the states belonging to one flux sector to another flux sector. Note however that the new flux sector is in the same conjugacy class as the original flux. More formally, we have
 \begin{eqnarray}
  \cA_{h^{-1}} B_{C_g} \cA_{h} & = & \sum_{f\in C_g} \cA_{h^{-1}} B_{f} \cA_{h} \\
  & = & \sum_{f\in C_g} B_{h^{-1} f h} \\
  & = & B_{C_g}
 \end{eqnarray}
 The commutation relation~\eqref{eq:commutation} follows by noting that $\cA_{h^{-1}}=(\cA_{h})^{-1}$ since vertex operators are a representation of $G$.
\end{proof}

The immediate corollary is that charge projectors also commute with flux projectors since they are linear combination of vertex operators.

\begin{cor}
\begin{equation}
 [A_{\Gamma^G}, B_{C_g} ] = 0
\end{equation}
 
\end{cor}

We can interpret the commutation of the 4-body projectors as a decoupling of the charges from the fluxes. However, there's an apparent catch with both this statement and this formalism: all the $A_{\Gamma^G}$ charge projectors project unto an irreducible representation of the full group $G$, rather than the appropriate normalizer subgroups $\cN_h$ to which the charges are actually assigned. This hints at the fact that excitations of distinct energy in our family of Hamiltonians are not precisely anyons. Indeed, the internal states of some anyon types will now be split into two different energy eigenspaces. We will see how this manifests itself on the example of $\mathcal{D}(S_{3})$, in Secs.~\ref{sec:Hilbert-splitting}-\ref{sec:local-dof}. Let's start by working out in details the flux and charge projectors of $\mathcal{D}(S_{3})$.

\subsection{Example of $G=S_{3}$}
\label{subsec:D(S_3)_Hamiltonian}

The flux projectors \eqref{eq:flux_projector} in the case of $G=S_{3}$ are:
\bq
B_{C_e} &=& B_e ,\\
B_{C_y} &=& B_y + B_{y^2} ,\\
B_{C_x} &=& B_x + B_{xy} + B_{xy^2} .
\eq

The 4-body charge projectors \eqref{eq:charge_projector} for $S_3$ are:
\bq
A_{\Gamma_1} &=& \frac{1}{6} (\mathcal{A}_e + \mathcal{A}_y + \mathcal{A}_{y^2} + \mathcal{A}_x + \mathcal{A}_{xy} + \mathcal{A}_{xy^2}) ,\\
A_{\Gamma_{-1}} &=& \frac{1}{6} (\mathcal{A}_e + \mathcal{A}_y + \mathcal{A}_{y^2} - \mathcal{A}_x - \mathcal{A}_{xy} - \mathcal{A}_{xy^2}) ,\\
A_{\Gamma_2} &=& \frac{1}{3} (2 \mathcal{A}_e - \mathcal{A}_y - \mathcal{A}_{y^2}) ,
\eq
since they are based on the characters of the irreducible representations of $S_3$ (see Table~\ref{tab:irreps_of_S3} for the irreps of $S_3$). 

The refined Hamiltonian \eqref{eq:massiveH_4body} is then:
\begin{eqnarray} \label{eq:DS3_Hamiltonian}
H & = & \sum_v (\alpha A^v_{\Gamma_1} + \beta A^v_{\Gamma_{-1}} + \gamma A^v_{\Gamma_2} ) \nonumber \\ 
& & + \sum_p (\delta B^p_{C_e} + \epsilon B^p_{C_x} + \nu B^p_{C_y} ) .
\end{eqnarray}

In contrast, the 6-body projectors \eqref{eq:6-body_projector} have the form (see Table~\ref{tab:D(S3)_anyons} for the labeling of anyons of $\mathcal{D}(S_3)$):
\bq
\label{eq:6-body_proj_A}
P_A &=& \frac{1}{6} (\mathcal{A}_e + \mathcal{A}_y + \mathcal{A}_{y^2} + \mathcal{A}_x + \mathcal{A}_{xy} + \mathcal{A}_{xy^2}) B_{C_e} , \\
P_B &=& \frac{1}{6} (\mathcal{A}_e + \mathcal{A}_y + \mathcal{A}_{y^2} - \mathcal{A}_x - \mathcal{A}_{xy} - \mathcal{A}_{xy^2}) B_{C_e} , \\
P_C &=& \frac{1}{3} (2 \mathcal{A}_e - \mathcal{A}_y - \mathcal{A}_{y^2}) B_{C_e} , \\
P_D &=& \frac{1}{2} (\mathcal{A}_e + \mathcal{A}_x) B_{x} + \frac{1}{2} (\mathcal{A}_e + \mathcal{A}_{xy}) B_{xy} \nonumber \\
&& + \frac{1}{2} (\mathcal{A}_e + \mathcal{A}_{xy^2}) B_{xy^2} , \\
P_E &=& \frac{1}{2} (\mathcal{A}_e - \mathcal{A}_x) B_{x} + \frac{1}{2} (\mathcal{A}_e - \mathcal{A}_{xy}) B_{xy} \nonumber \\
&& + \frac{1}{2} (\mathcal{A}_e - \mathcal{A}_{xy^2}) B_{xy^2} , \\
P_F &=& \frac{1}{3} (\mathcal{A}_e + \mathcal{A}_y + \mathcal{A}_{y^2}) B_{C_y} , \\
P_G &=& \frac{1}{3} (\mathcal{A}_e + \omega \mathcal{A}_y + \bar{\omega} \mathcal{A}_{y^2}) B_{C_y} , \\
P_H &=& \frac{1}{3} (\mathcal{A}_e + \bar{\omega} \mathcal{A}_y + \omega \mathcal{A}_{y^2}) B_{C_y} .
\label{eq:6-body_proj_H}
\eq

The corresponding 6-local Hamiltonian \eqref{eq:massiveH_6body} would allow to freely tune the masses of the anyons, albeit at a cost of a more non-local Hamiltonian.

\section{Hilbert space splitting}
\label{sec:Hilbert-splitting}

In this Section, we elaborate on the way the charge and flux projectors split up the Hilbert space of a site. Indeed, we will see in Sec.~\ref{sec:consistent-splitting} that both the charge and flux family of projectors provide a distinct way to split the Hilbert space unto which they are acting non-trivially. Moreover, since those projectors commute, those two splittings are consistent over the Hilbert space unto which they both act non-trivially, i.e. the Hilbert space of 2 spins which has dimension $|G|^2$. 

We will argue that the splitting of the common Hilbert space of charge and flux operators induces a splitting of the proper Hilbert space of a site. Because sites overlap, the dimension of the proper Hilbert space of a single site is smaller than the Hilbert space of the 6 spins forming the site. We prove in Sec.~\ref{sec:dimension-site} that this proper Hilbert space also has dimension $|G|^2$.  In Sec.~\ref{sec:diagram}, we introduce a diagrammatic representation of this splitting. This diagram encapsulates all the results of this paper about the structure of refined quantum double models.

\subsection{Two distinct yet consistent ways to split the Hilbert space}
\label{sec:consistent-splitting}

We first prove that the charge and flux projectors, which respectively act non-trivially on four spins, add up to the identity operator on the Hilbert space of dimension $|G|^4$ of the four spins. Since they are orthogonal projectors, charge (resp. flux) projectors provide an orthogonal resolution of the identity, i.e., the direct sum of their images amounts to the full Hilbert space.

\subsubsection{Resolution of the identity for charge projectors}

\begin{lem}
The dimension of the image of the charge projector for the irreducible representation $\Gamma$ is
\begin{equation}
\mbox{Tr}\left[A_{\Gamma}\right]=\left|G\right|^{3}d_{\Gamma}^{2}
\end{equation}
where $d_\Gamma$ is the dimension of the irrep $\Gamma$.
\end{lem}

\begin{proof}
Recall that the vertex operators $\cA_g$ are tensor products of 4 copies of the (left) regular
representation $L$. $L(g)$ matrices are permutations with no fixed points, unless $g=e$. Since the trace of a tensor product is the product of the trace, $A(g)$ is traceless unless $g=e$. The vertex operator $\cA_e$ is nothing but the identity matrix on a space of dimension $|G|^4$. Thus, 
\begin{equation}
 \mbox{Tr}\cA_{g}=\left|G\right|^{4}\delta_{ge}
\end{equation}

Simple calculation yields
\begin{eqnarray}
\mbox{Tr}\left[A_{\Gamma}\right] & = & \frac{d_{\Gamma}}{\left|G\right|}\sum_{g\in G}\chi_{\Gamma}(g)\mbox{Tr} \cA_{g} \nonumber \\
 & = & \left|G\right|^{3}d_{\Gamma}\chi_{\Gamma}(e) \nonumber \\
 & = & \left|G\right|^{3}d_{\Gamma}^{2} .
\end{eqnarray}
\end{proof}

To see that the charge projectors add up to the identity on the Hilbert space of the 4 spins, we use a well-known fact from representation theory
\begin{equation}
 \sum_{\Gamma}d_{\Gamma}^{2}=|G|.
\end{equation}
Dimension counting
and the fact that charge projectors are orthogonal allows us to conclude that  
\begin{equation}
\sum_{\Gamma} A_{\Gamma}=\id_{|G|^{4}}
\end{equation}
i.e., the charge projectors are an orthogonal resolution of the identity for the Hilbert space of the 4 spins neighboring a vertex.

\subsubsection{Resolution of the identity for flux projectors}

\begin{lem}
The dimension of the image of the flux projector for the conjugacy class $C_g \subset G$ is
\begin{equation}
\mbox{Tr}\left[B_{C_g}\right]=|C_g||G|^{3}.
\end{equation}
where $|C_g|$ is the cardinality of the conjugacy class.
\end{lem}

\begin{proof}
Flux projectors are sum of rank-one projectors unto fluxes that belong
to the same conjugacy class $C_g$. Thus, to compute the dimension
of the image of the flux projectors, one needs to compute how many
terms appear in the sum, i.e., how many ways 4 group elements can
be multiplied such that their product belongs to the conjugacy class
$C_g$. The first three group elements $a,b,c$ can be chosen
arbitrarily in $|G|^{3}$ distinct ways. Then the fourth
group element $d$ is chosen such that the product belongs to the
conjugacy class $C_g$, i.e., $d\in(abc)^{-1} C_g$. Thus,
there are $|C_g|$ choices for $d$. This concludes the proof.
\end{proof}

Moreover, since every group element belongs to one and only one conjugacy
class, we know that 
\begin{equation}
 \sum_{C_g\subset G}\left|C_g\right|=|G|.
\end{equation}
Dimension counting and the fact that flux projectors are orthogonal
allows us to conclude that 
\begin{equation}
\sum_{C_g\subset G}B_{C_g}=\id_{|G|^{4}}
\end{equation}
i.e., the flux projectors are an orthogonal resolution of the identity for the Hilbert space of the 4 spins of a plaquette.

\subsection{Dimension of the proper Hilbert space of a site}
\label{sec:dimension-site}

Since the flux and charge projectors pairwise commute (see Sec.~\ref{subsubsec:Commutation}), they provide a consistent splitting of the Hilbert space unto which they both act non-trivially in the sense that a basis of this Hilbert space is spanned by common eigenstates. It is clear that the intersection of their geometric support is two spins. The corresponding Hilbert space has dimension $|G|^2$.

Here we want to argue that this splitting of Hilbert space induces a splitting of the Hilbert space of a site. Naively, a site is made of 6 spins but since spins are shared by many sites, the dimension of its proper Hilbert space is smaller than $|G|^6$. We will show that it is $|G|^2$, the same as the common Hilbert space of flux and charge projectors.

To determine the dimension of this proper Hilbert space, first recall that a site is the union of the four spins around a plaquette and the four spins around a neighboring vertex. Since 2 spins are shared, a site consists of 6 spins. However, each spin belongs to three distinct sites: one site in which it belongs to both the vertex and the plaquette, one site for the other vertex and one site for the other plaquette, see Fig.~\ref{fig:1edge3sites}. Thus, the dimension of the (proper) Hilbert space associated to every site is 

\begin{equation}
 d\left( \cH_\mathrm{site} \right) = \sqrt[3]{|G|^6} = |G|^2
\end{equation}

\begin{figure}[ht]
\begin{centering}
\includegraphics[width=0.25\textwidth]{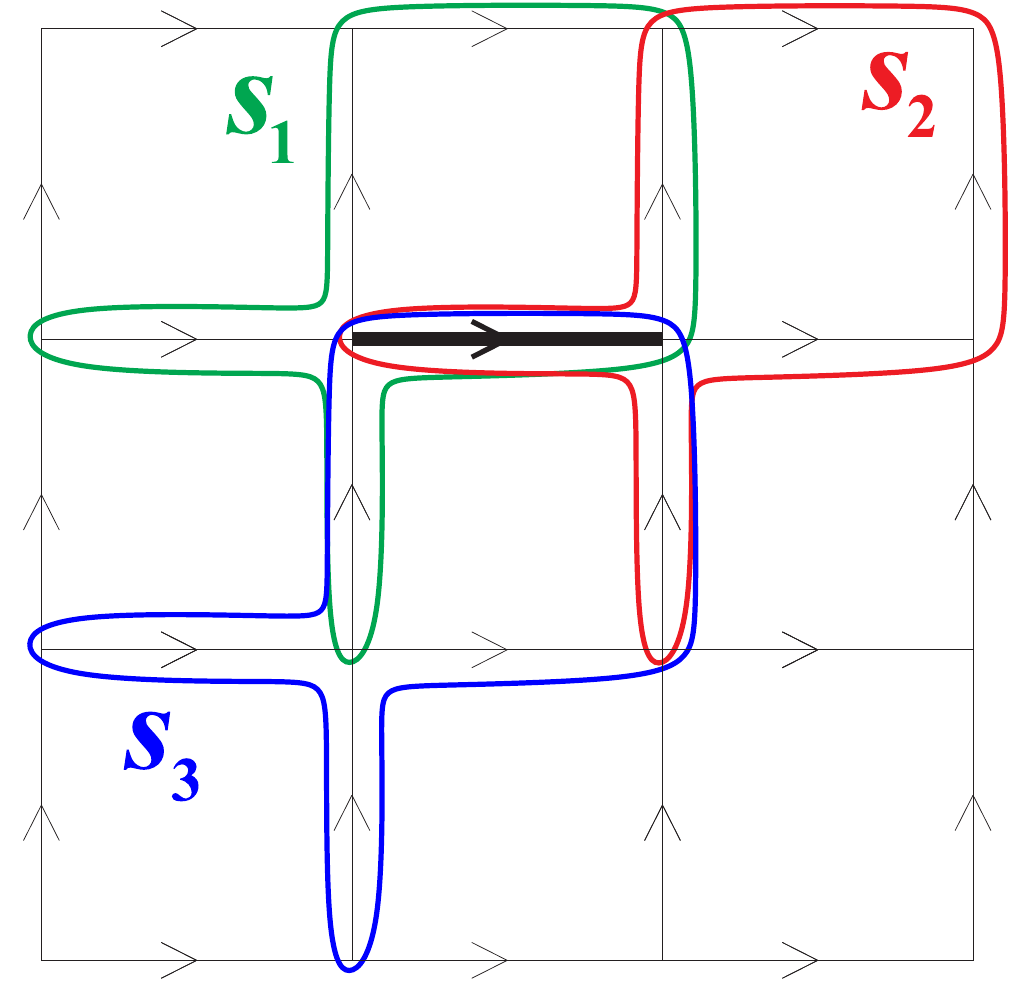}
\caption{(Color online) Illustration of the fact that every edge belongs to exactly 3 sites. For the thick edge in the figure the 3 sites are $s_1$, $s_2$ and $s_3$.}
\label{fig:1edge3sites}
\end{centering}
\end{figure}

A simple way to think about this is that for every site, the two spins shared between the vertex and the plaquette are assigned to this site while other spins of the site are assigned to other neighboring sites.

\subsection{Diagrammatic representation and energy sectors}
\label{sec:diagram}

We now introduce a diagrammatic representation of the splitting of the proper Hilbert space of a site which we consider to be a very useful tool to better understand the structure of quantum double models.

The diagram, represented on Fig.~\ref{fig:anyon_partitioning} for the case of $\mathcal{D}(S_3)$, is a square of size $|G|$. Each column is indexed by an irrep $\Gamma$ of $G$ and its width is the squared dimension of the irrep $d^2_\Gamma$. Columns thus correpond to the splitting of the Hilbert space induced by the charge projectors. Similarly,  each row is indexed by a conjugacy class $C_g$ of $G$ and its width is the cardinality of the conjugacy class $|C_g|$. Rows correspond to the splitting induced by the flux projectors.

\begin{figure}[ht]
\begin{centering}
\includegraphics[width=0.8\columnwidth]{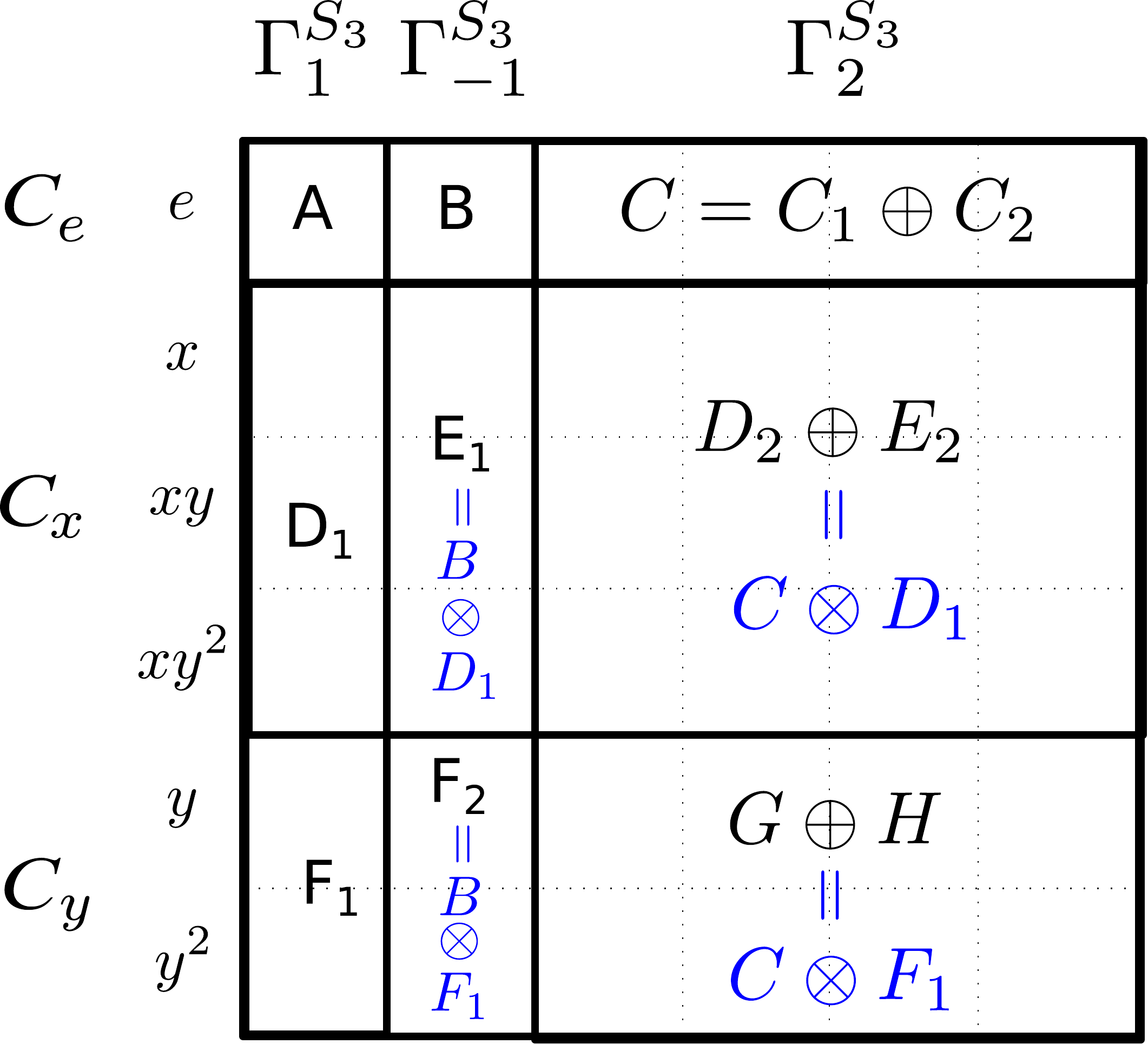}
\caption{(Color online) The flux and charge projectors of $\mathcal{D}(S_3)$ partition the Hilbert space of dimension $|S_3|^2=36$ unto which both family of operators act non-trivially. The charge projector splitting defines columns. The flux projectors, corresponding to conjugacy classes, define rows (each row between dotted lines corresponds to a group element). The 9 energy sectors are represented $\{ A, B, C=C_1 \oplus C_2, D_1, E_1, D_2 \oplus E_2, F_1, F_2, G \oplus H  \}$. The labels are chosen to reflect the relation of the excitations with the 8 anyon types of $\mathcal{D}(S_3)$. In particular, anyons $D$, $E$ and $F$ appear in two distinct energy sectors and there are two copies of the  chargeon $C$, labeled $C_1$ and $C_2$. Note that the area of the surface attributed to each anyon is equal to the square of its quantum dimension. Labels in blue correspond to the reinterpretation of certain excitations as a combination of other excitations (see Sec.~\ref{subsubsec:flavor_interpretation}). }
\label{fig:anyon_partitioning}
\end{centering}
\end{figure}

\subsubsection{Labeling of the energy sectors}

Each intersection is now labeled by a conjugacy class and an irrep. Notice that every such intersection will have a well-defined energy (see our Hamiltonian, Eq.~\eqref{eq:DS3_Hamiltonian}), we will call these intersections \emph{energy sectors}. However, these energy sectors do not correspond directly to anyon types since an irrep of the full group $G$ can split into the direct sum of irreps of the normalizer of the conjugacy class. Let's explore this on the example of $G=S_3$. 

\paragraph{Trivial representation: $D_1$ and $F_1$}

Restricting the trivial representation of $S_3$ to the normalizer subgroup $\cN_x$ or $\cN_y$ will correspond to the trivial representation of both of those subgroups, i.e., 

\begin{align}
 \Gamma^{S_3}_1 |_{\cN_x} & =  \Gamma^{\mathbb{Z}_2}_1 \label{eq:splitup_Gamma1_Nx} \\
 \Gamma^{S_3}_1 |_{\cN_y} & =  \Gamma^{\mathbb{Z}_3}_1 \label{eq:splitup_Gamma1_Ny}
\end{align}
Thus, the energy sectors in the first column, corresponding to the trivial irrep of $S_3$, correspond to anyon types $A$, $D$ and $F$ (please refer to Table~\ref{tab:D(S3)_anyons} for anyon labels for $S_3$) depending on their row, i.e., their conjugacy class. For the non-abelian anyons $D$ and $F$, we will label those energy sectors $D_1$ and $F_1$ since we will see shortly that other energy sectors correspond to those anyon types as well.

\paragraph{Alternating representation: $E_1$ and $F_2$}

Similarly, restricting the alternating representation of $S_3$ to $\cN_x$ corresponds to the alternating representation of $\cN_x$ (excitation $E_1$), and restricting it to $N_y$ will give the trivial representation of $\cN_y$ (excitation $F_2$). Thus, we have already uncovered two energy sectors for anyon $F$.

\begin{align}
 \Gamma^{S_3}_{-1} |_{\cN_x} & =  \Gamma^{\mathbb{Z}_2}_{-1} \label{eq:splitup_Gamma-1_Nx} \\
 \Gamma^{S_3}_{-1} |_{\cN_y} & =  \Gamma^{\mathbb{Z}_3}_1 \label{eq:splitup_Gamma-1_Ny}
\end{align}

\paragraph{Two-dimensional representation: $D_2 \oplus E_2$ and $G \oplus H$}

Finally, the two-dimensional representation, restricted to $\cN_x$ or $\cN_y$ will break up to two 1-dimensional representations on the subgroups. These 1-dimensional representations will be the trivial and the alternating of $\cN_x$ (dyons $D_2$, $E_2$), and the two nontrivial representations of $\cN_y$ (dyons $G$ and $H$).

\bq \label{eq:splitup_Gamma2_Nx}
 \Gamma^{S_3}_2 |_{\cN_x} & =&  \Gamma^{\mathbb{Z}_2}_1 \oplus \Gamma^{\mathbb{Z}_2}_{-1} \\
 \Gamma^{S_3}_2 |_{\cN_y} & =&  \Gamma^{\mathbb{Z}_3}_{\omega} \oplus \Gamma^{\mathbb{Z}_2}_{\bar{\omega}}
\label{eq:splitup_Gamma2_Ny}
\eq

We refer the reader to Tables~\ref{tab:irreps_of_S3}-\ref{tab:irreps_of_Z3_Z2} to check these relations between the representations of $S_3$ and its subgroups. The breaking of irreps of the full group $G=S_3$ into irreps of its subgroups $\mathbb{Z}_2$ and $\mathbb{Z}_3$ explains how our refined Hamiltonian correctly accounts for anyons D, E, F, G and H although the irreps of the normalizers $\mathbb{Z}_2$ and $\mathbb{Z}_3$ do not have an associated Hamiltonian term. This property can be made general for an arbitrary group $G$ by discussing \emph{induced} representations, which we do in Appendix~\ref{sec:D(G)_splitup}.

Even in the smallest non-abelian example (quantum double of $S_3$), irrep breaking leads to a very intricate splitting of the Hilbert space. Consider the rectangle labeled by $C_x$ and $\Gamma^{S_3}_2$. The two-dimensional irrep will split into the sum of two one-dimensional irreps of $\mathbb{Z}_2$. However, the splitting is slightly different since the normalizers $\cN_x$, $\cN_{xy}$ and $\cN_{xy^2}$, while isomorphic, are not equal.

\subsubsection{Energies}

Recall the Hamiltonian given by Eq.~\eqref{eq:DS3_Hamiltonian}. We can compute the energy associated to each energy eigenspace (energy sector), which we denote $J$ to avoid confusion with anyon type $E$. This Hamiltonian assigns the following energies to the excitations: $J_A=\alpha+\delta$, $J_B=\beta+\delta$, $J_{C_1}=J_{C_2}=\gamma+\delta$, $J_{D_1}=\alpha+\epsilon$, $J_{D_2}=J_{E_2}=\gamma+\epsilon$, $J_{E_1}=\beta+\epsilon$, $J_{F_1}=\alpha+\nu$ and $J_{F_2}=\beta+\nu$, $J_G=J_H=\gamma+\nu$. Thus, we see that anyon types $D$, $E$ and $F$, can be in different energy eigenspaces. This is surprising and should not be possible from a topological point of view. However, anyons in a quantum double are not fundamental particles, rather emergent quasi-particles on a \emph{lattice} model. We will now explore further this discrepancy and see that the existence of local degrees of freedom on a lattice explains the different energies attributed to states corresponding to the same anyon at the mesoscopic level.

\subsubsection{Dimension and area of the diagram}

Finally, note that the area of the rectangle (or the sum of the areas of distinct rectangles when an anyon occupies different energy sectors) is exactly the squared quantum dimension of that anyon $(d_k)^2$. Since the area of the whole square is $|G|^2$, we recover the well-known result
\begin{equation}
 \mathcal{D}^2 \equiv \sum_k (d_k)^2 = |G|^2 \label{eq:area-diagram}
\end{equation}
We will see that the topological degrees of freedom of an anyon have dimension $d_k$ while the local degrees of freedom have also dimension $d_k$, which results in a dimension $(d_k)^2$ for each anyon.

\section{Local degrees of freedom}
\label{sec:local-dof}

We now elucidate the fact that anyon types are not in one-to-one correspondence with energy sectors. We will argue that anyon types are labels that are topological at the mesoscopic level, in the sense that they cannot be changed locally. However, additional \emph{local} degrees of freedom, which can be modified by local unitary transformations acting close to the excitations, also arise. We explore the complex interplay of those different types of degrees of freedom.

\subsection{Disagreement between anyons and energy sectors}

The way the Hilbert space of a site is split up by the charge and flux projectors, detailed in Sec.~\ref{sec:Hilbert-splitting}, leads to a disagreement between energy sectors of our Hamiltonian and anyon labels. Here, we will explain in detail what we mean by this disagreement. 

First, chargeons appear in mutliple copies. For $G=S_3$, the chargeon $C$ corresponding to the non-trivial 2D irrep appears in 2 copies, labelled $C_1$ and $C_2$. In general, an irrep $\Gamma^G$ will result in a number of copies equal to its dimension $d_{\Gamma^G}$. This simply reflects that the multiplicity of the irrep in the regular representation is equal to its dimension. 

Second, some anyons appear in multiple energy sectors. As an example, let's look at anyon D, which appears in two distinct energy sectors of the diagram since the trivial irrep of $\mathbb{Z}_2$ can be obtained from the trivial irrep of $S_3$, see Eq.~\eqref{eq:splitup_Gamma1_Nx}, or from the two-dimensional irrep of $S_3$, see Eq.~\eqref{eq:splitup_Gamma2_Nx}. We say that anyon D comes in two distinct \emph{charge flavors}. Each charge flavor is an eigenspace of the Hamiltonian. $D_1$ labels a subspace with dimension three and is within the image of the trivial irrep of $S_3$ whereas the label $D_2$ labels a subspace of dimension six and is within the image of the two-dim irrep of $S_3$. The same phenomenon relates $E_1$ to $E_2$ and $F_1$ to $F_2$.

It seems peculiar that a local observable allows to distinguish two subspaces of internal states of anyon $D$ (i.e. the two charge flavors). This even seems like a violation of anyonic properties of D, since by simply applying the local operators $A_{\Gamma^{S_3}_1}$ and $A_{\Gamma^{S_3}_2}$ we can establish a global labeling that differentiates between the two charge flavors based on their energies. How is that possible if both those charge flavors of D are just subspaces of one and the same anyon? We will argue that the anyon labelling corresponds to degrees of freedom that cannot be changed locally whereas there exist \emph{local degrees of freedom} that can be changed locally. The charge projectors discriminates among those local degrees of freedom.

The surprising property that site excitations corresponding to the same anyon type can have different energies is not a peculiarity of our family of Hamiltonians. In fact, this property was already present in Kitaev's original Hamiltonian. Indeed, in the original quantum double construction, the pairs $(D_1,\,D_2)$ and $(F_1,\,F_2)$ would have different energy. Our family of Hamiltonian simply highlights this property.

\subsection{The role of finite lattice spacing}

The charge projectors act on the four spins around a vertex, and not on the remaining two spins of a site, see Fig.~\ref{fig:A_B_acting_on_sites}. They can be interpreted as operators that coherently move all fluxon types and check that they transform according to the correct irrep~\cite{BM08}. In particular, note that those test fluxons do not enclose the flux content of the \emph{site}.

Now, let us recall the interference experiment described in Sec.~\ref{subsec:aharonov-bohm}, that allows us to determine the charge of a dyon by having test fluxons undergo a double slit experiment with the dyon located behind the slits. It was key in that experiment that the flux of the test fluxon ($a$) and the flux of the measured dyon ($b$) have commuting labels, i.e.,  $ab = ba$, in order to have interference. The consequence of this requirement was that the charge of the measured dyon was labeled by an irrep of the normalizer $\mathcal{N}_b$ rather than an irrep of the full group $G$. However, this requirement stemmed from the fact that in a topological quantum field theory, to determine the charge of a dyon, one cannot avoid enclosing the flux of the dyon as well. But does that fact still hold in our \emph{lattice} model?

Indeed, in the quantum double construction, the charge of a dyon is located on a vertex whereas its flux is located on a plaquette, see Fig.~\ref{fig:charge_flux_separation}. In other words, the lattice separates charge and flux. This separation then allows something that would be impossible in a field theory: to braid the test fluxon with the charge part of a dyon without enclosing its flux. The corresponding wordline for the text fluxon is represented in purple on Fig.~\ref{fig:charge_flux_separation} (worldline 1), whereas the wordline allowed by field theory is represented in black (worldline 2). Consequently, this experiment discriminates different \emph{charge flavors} of a dyon.

\begin{figure}[ht]
\begin{centering}
\includegraphics[width=1.0\columnwidth]{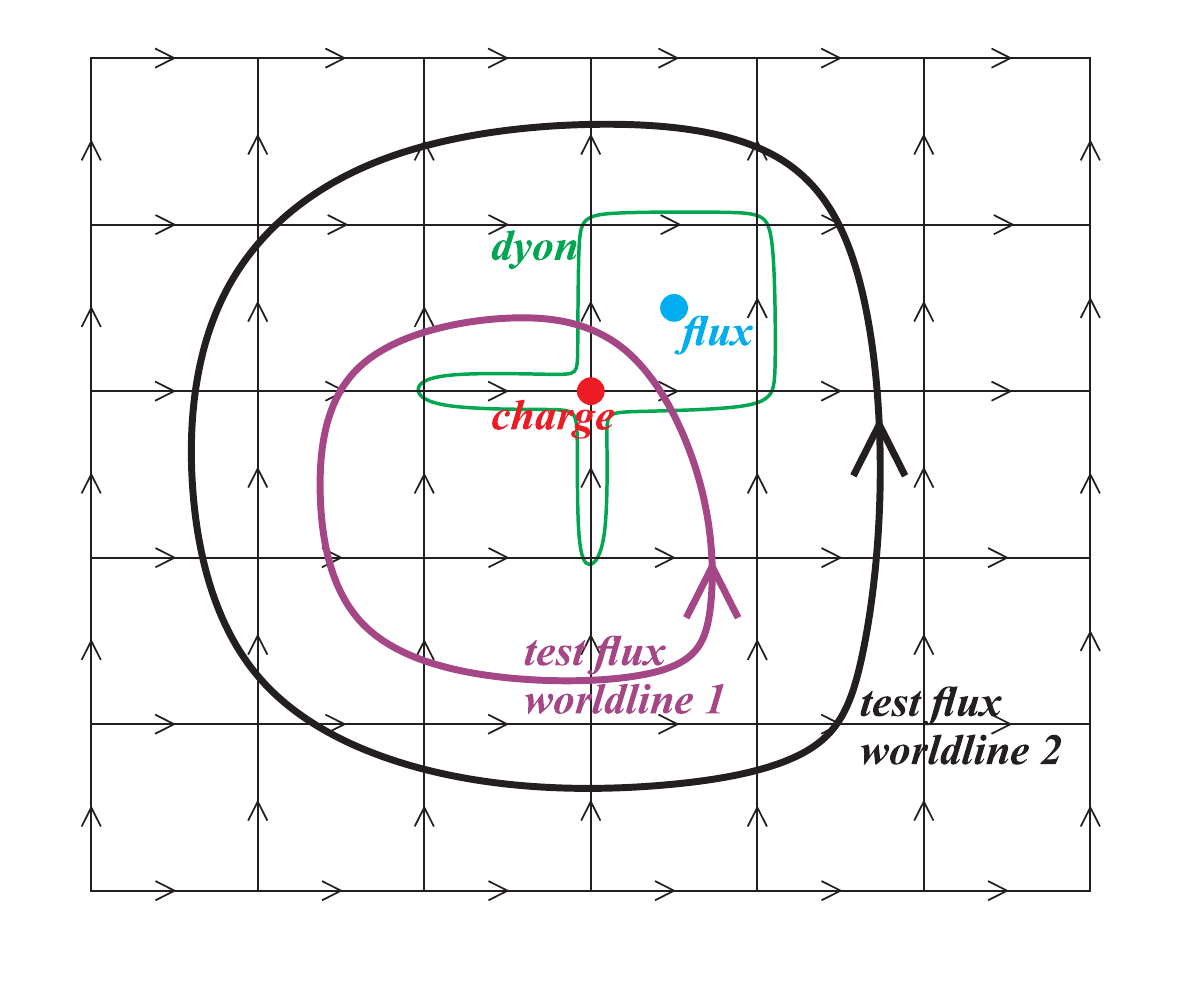}
\caption{Spatial separation of charge and flux of a dyon on the lattice: the charge is located on a vertex, while the flux is on the plaquette. This allows one to take a test flux around \emph{only the charge part} of a dyon following the test fluxon worldline 1. Such interferometric experiment allows to determine not only its charge but also its charge flavor since it is unaffected by the flux of the dyon. On the contrary, topological quantum field theory only allows for test fluxon worldline 2 which encloses both the charge and flux of the dyon. }
\label{fig:charge_flux_separation}
\end{centering}
\end{figure}

\subsubsection{Interpreting charge flavors in $\mathcal{D}(S_3)$}
\label{subsubsec:flavor_interpretation}

Based on the arguments above, we can understand better the meaning of the charge flavors $D_1$ vs $D_2$, $E_1$ vs $E_2$  or $F_1$ vs $F_2$. For instance, for anyon F: $F_1$ is a pure fluxon since its charge is the trivial irrep of $S_3$ while $F_2$ has the non-trivial alternating charge (see Fig.~\ref{fig:anyon_partitioning}). They both correspond to anyon $F$ since the alternating charge becomes trivial when restricted to the normalizer $\mathcal{N}_y$ as indicated by Eq.~\eqref{eq:splitup_Gamma-1_Ny}. One way to interpret this result is that $F_2$ is an excitation on a site which contains both a fluxon $F_1$ on the plaquette and a chargeon $B$ on the vertex. Thus,
\be
F_2 = B \otimes F_1  
\ee
which agrees with the known fusion rules of $\mathcal{D}(S_3)$ which state that $B \otimes F = F$ \cite{BSW11}. In terms of masses, one can notice that $M_{F_2}=M_{F_1}+M_B$, where for anyon X: $M_X = J_X - J_A$, i.e. mass is the energy penalty of anyon X compared to vacuum. Thus, one can think of $F_2$ as a composite anyon made of $F_1$ and $B$. We will see that a similar relation between masses will be true for the following examples as well.

Similarly, the anyon $E_1$ is a composite anyon made of the fluxon $D_1$ with the chargeon $B$
\be
E_1 = B \otimes D_1,
\ee
as well as the fusion rules state $B \otimes D = E$ \cite{BSW11}.

The other case of energy sector--anyon disagreement is slightly more involved since it involves a direct sum:
\begin{equation}
 D_2 \oplus E_2 = C \otimes D_1  
\end{equation}
i.e., the combination of the chargeon $C$ (either from the $C_1$ or $C_2$ copies) with a fluxon $D_1$ is a superposition of anyon $D$ and $E$ with a charge corresponding to the two-dimensional irrep of $S_3$. The Hamiltonian doesn't distinguish $D_2$ from $E_2$ since they have the same energy. The fusion rules are again in agreement with this statement: $C \otimes D = D \oplus E$~\cite{BSW11}. Similarly, the energy sector $G\oplus H$ results from the combination of a chargeon $C$ and a fluxon $F_1$:
\begin{equation}
 G \oplus H = C \otimes F_1,  
\end{equation}
which agrees with the fusion rule $C \otimes F = G \oplus H$ \cite{BSW11}.

We can relabel the energy sectors of Fig.~\ref{fig:anyon_partitioning} based on those combinations of flux and charge. The new labels are indicated in blue.

\subsection{Local vs global degrees of freedom}

We now argue that charge flavor is a local degree of freedom which can be transformed by a local unitary whereas anyon labels cannot be changed locally. We present an intuitive argument and refer to~\cite{BM08} for a formal, yet distinct, argument.

The charge flavor cannot be discriminated by any operator that encloses the whole site $s=(v,p)$, since it requires enclosing the vertex $v$ without enclosing the plaquette $p$ (see Fig.~\ref{fig:charge_flux_separation}). This means that two distinct charge flavors, say $D_1$ and $D_2$, have the same reduced density matrix outside the site, i.e., on the set of spins that do not belong to that site $s$. Yet, they correspond to distinct global states $\ket{\psi_1}$ and $\ket{\psi_2}$ on the whole lattice. Since those states are purifications of the same reduced density matrix, there exists a local unitary transformation $U_s$ acting only on the site $s$ such that $U_s \ket{\psi_1} = \ket{\psi_2}$. A similar statement holds for different copies of a chargeon such as $C_1$ and $C_2$.

The presence of local degrees of freedom explains that the dimension of the subspace associated with an anyon labeled by the conjugacy class $C_g$ and the irrep $\Gamma$ of its normalizer is 
\begin{equation}
 d^2=\left(|C_g||d_\Gamma|\right)^2
\end{equation}
rather than $d$, which we expect from topological quantum field theory~\cite{Preskill98}. The dimension of the anyon is the product of the dimensions of its local and topological degrees of freedom. It turns out that for quantum double models there are as many local degrees of freedom as global topological degrees of freedom~\cite{BM08}, i.e.,
\begin{equation}
 d_\mathrm{local}=d_\mathrm{topo}=d. 
\end{equation}

Thus, in a quantum double model, due to the lattice, each anyon corresponds to a subspace of dimension 
\begin{equation}
 d_\mathrm{local} \times d_\mathrm{global} = d^2.
\end{equation}
This result confirms the observation made on the anyon splitting diagram of Fig.~\ref{fig:anyon_partitioning} in which each anyon corresponds to a surface of area $d^2$. Moreover, the total area $|G|^2=\sum_k d_k^2$ is a graphical representation of the identity given by Eq.~\eqref{eq:area-diagram}.

\section{Discussion}
\label{sec:conclusions}

In this paper, we introduced a new family of 2D topological spin lattice models which generalize Kitaev's quantum double construction. The Hamiltonian of this new class of topological models are given by a translation-invariant sum of local commuting terms acting each on 4 neighboring spins. 

We provided a proof on the commutation of those operators which is based on a basis-independent reformulation of the Great Orthogonality Theorem.

Each local term of that refined Hamiltonian can be multiplied by a coupling constant which makes the energy spectrum of those models richer than the original Kitaev quantum double construction. Moreover, the new Hamiltonian highlights the feature that point-like excitations on a site corresponding to the same anyon can have different energies. This feature arises because the lattice introduces local degrees of freedom in addition to topological degrees of freedom. The interplay between those degrees of freedom might lead to surprising consequences.

\subsection{Consequences for quantum computation}

The disagreement between anyons and energy sectors is already present in the original quantum double construction, since Kitaev's Hamiltonian would give different masses to $D_1$ which is a fluxon than $D_2$ which is a dyon (from the point of view of irreps of $S_3$). Similar properties hold for the two charge flavors of anyon $F$, labeled $F_1$ and $F_2$, as well as anyon $E$, labeled $E_1$ and $E_2$, the latter would however not be distinguished by Kitaev's Hamiltonian. This leads us to the troubling question of what (if any) consequences will arise in quantum computation with non-abelian anyons when performing them on a lattice?

As the disagreement between anyons and energy sectors arises due to the finite separation between flux and charge of a dyon, one would have to be careful to perform every braiding procedure on a large scale, making sure to always braid with both flux and charge of a dyon. On a large enough lattice system, we can imagine the spacing will become insignificant, and no consequences will arise. On the other hand, the environment could introduce local noise that will project out one or the other charge flavor of a dyon, possibly resulting in unexpected processes, if for example, the local degrees of freedom entangle with the topological degrees of freedom. It is possible that this will not create problem for topological quantum computation since it occurs in fusion space. Nonetheless, clarifying those consequences needs careful consideration, and is the scope of future work.

\subsection{Consequences for quantum memories}

Using our family of Hamiltonians allows for tuning the masses of excitations, which will modify both the coherent dynamics and the incoherent dynamics of the topological model in the presence of a (thermal) environment. Thus, our family of Hamiltonian opens a new possibility for quantum self-correcting models based on topological models. Indeed, our models generalize the abelian construction in Ref.~\cite{BAP14} where a parameter regime interesting for quantum self-correction was identified. In that regime, it was argued that entropic effects lead to a different scaling of the memory time. While that improvement was shown to not carry over in the low temperature regime~\cite{KLT16}, a non-abelian model might yield a different result or, at least, allow for a better understanding of entropic effects in quantum double models.

\subsection{Holography between local, topological and fusion degrees of freedom?}

The fact that local degrees of freedom and topological degrees of freedom have the same dimension $d_k$ (where $k$ labels the anyon types) might be a clue pointing to an underlying holography. Moreover, the dimension of the subspace associated to an anyon on a site is $(d_k)^2$, which is the same dimension as the fusion space of two anyons of type $k$. We wonder whether this also hints at a deeper mathematical/physical connection.

Finally, it seems that local degrees of freedom are somehow unavoidable in a quantum double construction. Indeed, anyons live on a site, whose proper Hilbert space dimension is $|G|^2$. In the absence of local degrees of freedom, the direct sum of every anyon subspace would have dimension $\sum_k d_k$. Since this last quantity is not simply related to the dimension of the group $|G|$, local degrees of freedom have to account for the dimension mismatch. The situation is very different in Levin-Wen models~\cite{LW05} in which the dimension of the spin is precisely the number of anyon types.

\section{Acknowledgments}

We thank John Preskill, Alexei Kitaev, David Aasen, Ben Levitan, Sujeet Shukla and Dominic Williamson for helpful discussions. We acknowledge funding provided by the Institute for Quantum Information and Matter, an NSF Physics Frontiers Center (NSF Grant PHY-1125565) with support of the Gordon and Betty Moore Foundation (GBMF-2644). OLC is partially supported by the Natural Sciences and Engineering Research Council of Canada (NSERC).

\clearpage
\appendix

\section{Mathematical proofs}

We now detail the mathematical proofs of Sec.~\ref{sec:tunable_Hamiltonian}.

\subsection{Proof of Lemma~\ref{lem:GOT-swap}}\label{subsec:proof-lemma-GOT}

To prove Theorem~\ref{thm:charge-projector}, we need to first prove Lemma~\ref{lem:GOT-swap} which is a restatement of the Great Orthogonality theorem, Fact~\ref{fact:GOT}.

\begin{lem*}[Basis-independent GOT]
\begin{equation}
\sum_{g\in G}\Gamma(g)\otimes\Lambda(g^{-1})=\frac{\left|G\right|}{d_{\Gamma}}\delta_{\Gamma\Lambda}S
\end{equation}
where $S$ is the swap operator, i.e., $S:\mathbb{C}^{d}\times\mathbb{C}^{d}\to\mathbb{C}^{d}\times\mathbb{C}^{d}$
is defined by $S\left(|i\rangle\otimes|j\rangle\right)=|j\rangle\otimes|i\rangle$.\end{lem*}

\begin{proof}
The proof is a sequence of simplifications. The GOT is specifically used to simplify Eq.~\eqref{eq:use-GOT}:

\begin{eqnarray}
\sum_{g\in G}\Gamma(g)\otimes\Lambda(g^{-1}) \\
=  \sum_{g\in G}\sum_{ij}\left(\Gamma(g)\right)_{ij}\ketbra{i}{j}\otimes\sum_{k\ell}\left(\Lambda(g^{-1})\right)_{k\ell}\ketbra{k}{\ell}\\
  =  \sum_{g\in G}\sum_{ij}\left(\Gamma(g)\right)_{ij}\ketbra{i}{j}\otimes\sum_{k\ell}\overline{\left(\Lambda(g)\right)_{\ell k}}\ketbra{k}{\ell}\\
  =  \sum_{ijk\ell}\sum_{g\in G}\left(\Gamma(g)\right)_{ij}\overline{\left(\Lambda(g)\right)_{\ell k}}\ketbra{i}{j}\otimes\ketbra{k}{\ell} \label{eq:use-GOT} \\ 
  =  \sum_{ijk\ell}\frac{|G|}{d_{\Gamma}}\delta_{\Gamma\Lambda}\delta_{i\ell}\delta_{jk}\ketbra{i}{j}\otimes\ketbra{k}{\ell}\\
  =  \frac{|G|}{d_{\Gamma}}\delta_{\Gamma\Lambda}\sum_{ij}\ketbra{i}{j}\otimes\ketbra{j}{i}\\
  =  \frac{\left|G\right|}{d_{\Gamma}}\delta_{\Gamma\Lambda}S
\end{eqnarray}

\end{proof}

\subsection{Proof of Theorem~\ref{thm:charge-projector}} \label{subsec:proof-orthonormality}

We can now prove Theorem \ref{thm:charge-projector}.

\begin{thm*}[Orthogonality of charge projectors]
The operators defined by Eq. \eqref{eq:charge_projector} are orthonormal projectors
\begin{equation}
A_{\Gamma}A_{\Lambda}=\delta_{\Gamma\Lambda}A_{\Gamma}
\end{equation}
\end{thm*}

\begin{proof}
Simple algebra shows that 
\begin{eqnarray}
A_{\Gamma}^s A_{\Lambda}^s & = & \frac{d_{\Gamma}d_{\Lambda}}{|G|^{2}}\sum_{g,g'\in G}\chi_{\Gamma}(g)\chi_{\Lambda}(g') \cA_{g}^s \cA_{g'}^s\\
 & = & \frac{d_{\Gamma}d_{\Lambda}}{|G|^{2}}\sum_{h\in G}\underbrace{\sum_{g\in G}\chi_{\Gamma}(g)\chi_{\Lambda}(g^{-1}h)}_{(*)} \cA_{h}^s
\end{eqnarray}

We thus would like to prove that the $(*)$ term is proportional
to $\delta_{\Gamma\Lambda} \cdot \chi_{\Lambda}(h)$. 

Using the fact that $\mbox{Tr}\left[A\otimes B\right]=\mbox{Tr}\left[A\right]\mbox{Tr}\left[B\right]$,
one can rewrite the $(*)$ term as
\begin{equation}
(*)=\mbox{Tr}\left[\left(\sum_{g\in G}\Gamma(g)\otimes\Lambda(g)^{\dagger}\right)\left(\mathbb{I}\otimes\Lambda(h)\right)\right] .
\end{equation}
We can now use Lemma \ref{lem:GOT-swap} to express the trace as 
\begin{eqnarray}
(*) & =\delta_{\Gamma\Lambda} & \frac{|G|}{d_{\Gamma}}\mbox{Tr}\left[\sum_{ij}(\ketbra{i}{j}\otimes\left(\ketbra{j}{i}\right)\Lambda(h)\right]\\
 & = & \delta_{\Gamma\Lambda}\frac{|G|}{d_{\Gamma}}\sum_{ij}\delta_{ij}\langle i|\Lambda(h)|j\rangle\\
 & = & \delta_{\Gamma\Lambda}\frac{|G|}{d_{\Gamma}}\sum_{i}\left(\Lambda(h)\right)_{ii}\\
 & = & \delta_{\Gamma\Lambda}\frac{|G|}{d_{\Gamma}}\chi_{\Lambda}(h)
\end{eqnarray}
which concludes the proof of Theorem \ref{thm:charge-projector}.
\end{proof}

\subsection{Proof of Lemma~\ref{lem:flux-permutation}}\label{subsec:proof-lemma-flux}

We prove Lemma~\ref{lem:flux-permutation}.

\begin{lem*}[Flux permutation by vertex operators] For a plaquette $p$ and vertex $v$ that form a site, $(p,v)=s$
 \begin{equation}
 B^{(p)}_g = \cA^{(v)}_{h^{-1}} B^{(p)}_{hgh^{-1}} \cA^{(v)}_h ; \label{eq:appendix_commutation-lemma} 
 \end{equation}

for a plaquette $p$ and vertex $v$ that are parts of different sites, $p \in s_1$, $v\in s_2$, $s_1 \neq s_2$
 \begin{equation}
 B^{(p)}_g = \cA^{(v)}_{h^{-1}} B^{(p)}_{g} \cA^{(v)}_h . \label{eq:appendix_commutation-lemma2} 
 \end{equation}

\end{lem*}

\begin{figure}
\begin{centering}
\includegraphics[width=0.5\textwidth]{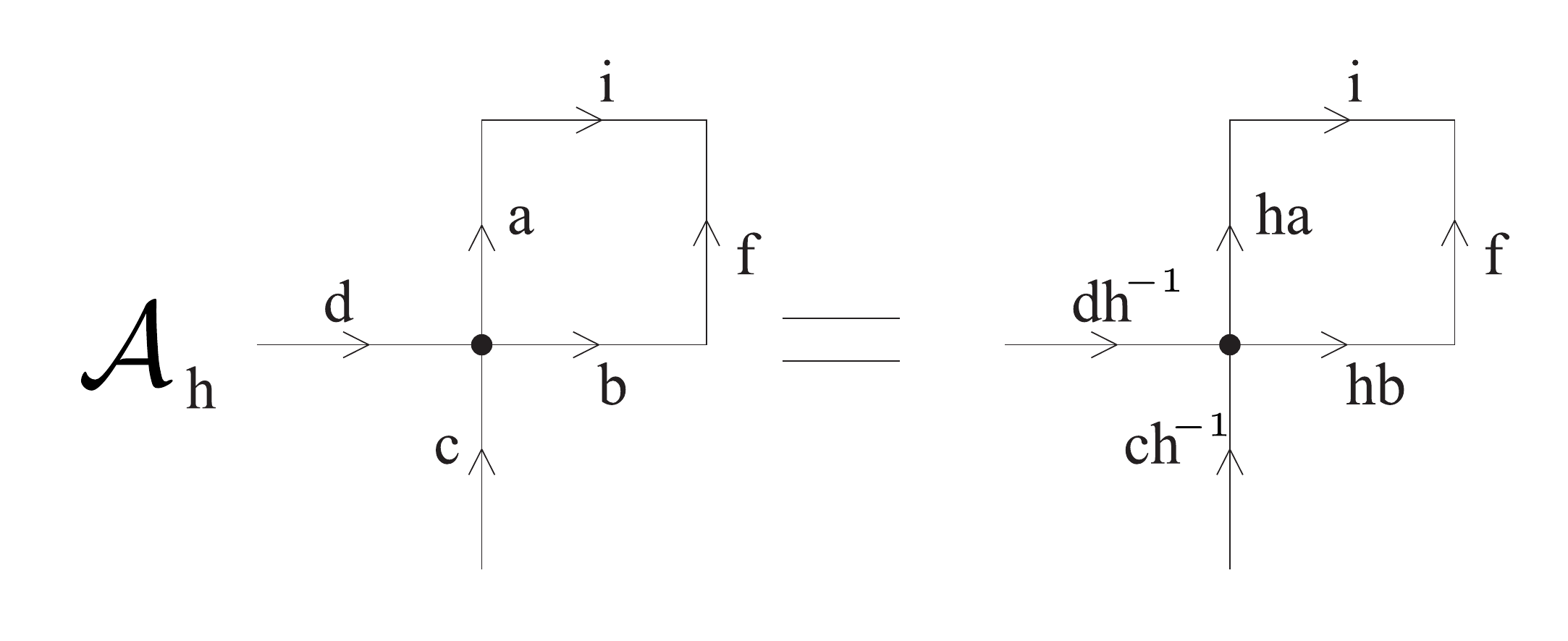}
\caption{Relative configuration of a vertex and a plaquette in the case when the commutation of charge and flux projectors is nontrivial. The figure shows how a vertex operator acts on these spins. Note that the flux around the plaquette, starting from the vertex, is $g=bfi^{-1}a^{-1}$ prior to the application of $\mathcal{A}_h$. Afterward, the flux is now $g'=hbfi^{-1}(ha)^{-1}=hgh^{-1}$.}
\label{fig:commutation}
\end{centering}
\end{figure}

\begin{proof}
 We will check the operator equality for an arbitrary state in which each spin is in a flux state (such states span the full (Hilbert) space).
 Note that the plaquette operator $B_g$ is in fact a projector unto states with flux $g$ threading the plaquette while states having a different flux are annihilated by $B_g$. Thus, the Hilbert space is split into a direct sum
 \begin{equation} \label{eq:direct-sum-proj}
  \mathcal{H}=I_g \oplus K_g 
 \end{equation}
 where $I_g$ (resp. $K_g$) denotes the image (resp. kernel) of the projector. The image is spanned by states with flux $g$ while states with other flux span the kernel. We will prove Eq.~\eqref{eq:appendix_commutation-lemma} first for a state in $I_g$ and then for a state in $K_g$.  
 
 For a state $\ket{\psi_g}$ whose flux is $g$, i.e., $B_g \ket{\psi_g}=\ket{\psi_g}$ the application of the vertex operator $\cA_h$ will act non-trivially on two spins around the plaquette and change its flux to $hgh^{-1}$ (when the plaquette and vertex operators act on the same site, see Fig.~\ref{fig:commutation}). Thus, $\cA_h\ket{\psi_g}$ is in the image of $B_{hgh^{-1}}$, i.e.,
 \begin{equation}
  \cA_h \ket{\psi_g} = B_{hgh^{-1}} \cA_h \ket{\psi_g}
 \end{equation}
 
 Finally, applying $\cA_{h^{-1}}$ will restore the spins into their original state and, in particular, restore the flux to $h^{-1}(hgh^{-1})h=g$, so that
 \begin{equation}
  \cA_{h^{-1}} B_{hgh^{-1}} \cA_h \ket{\psi_g} = \ket{\psi_g}.
 \end{equation}
 
 Let's now consider a state $\ket{\phi}$ whose flux is not $g$, i.e., $B_g \ket{\phi}=0$. That state is a linear combination of states with flux $f\neq g$. Let's assume that $\ket{\phi}$ has a well-defined flux $f$ (the general case will follow by linearity). Then,  $\cA_h \ket{\phi}$ will have flux $hfh^{-1}$ and will be annihilated by $B_{hgh^{-1}}$ since $hfh^{-1}\neq hgh^{-1}$. Thus, 
 \begin{equation}
  \cA_{h^{-1}} B_{hgh^{-1}} \cA_h \ket{\phi} = 0.
 \end{equation}
 
 Since we checked Eq.~\eqref{eq:appendix_commutation-lemma} on the two sectors of Eq.~\eqref{eq:direct-sum-proj}, it is valid for any state of the Hilbert space. Please note that we proved Eq.~\eqref{eq:appendix_commutation-lemma} only for one respective position of the vertex with respect to the plaquette. For the other three respective positions one can dutifully check that the proof is also valid, resulting in Eq.~\eqref{eq:appendix_commutation-lemma2}.
 
\end{proof}

\section{Induced representations of an arbitrary quantum double $\mathcal{D}(G)$}
\label{sec:D(G)_splitup}

A surprising feature of our refined quantum double Hamiltonian \eqref{eq:DS3_Hamiltonian} (see Eq.~\eqref{eq:massiveH_4body} for the general form) is that irreps of normalizers that are proper subgroups of $G$ do not have an associated Hamiltonian term. For instance, in the case of $\mathcal{D}(S_3)$, the irreps of $\mathbb{Z}_2$ and $\mathbb{Z}_3$ do not have an associated Hamiltonian term. How is it then that anyons D, E, F, G and H which are labelled by irreps of those two subgroups are correctly accounted for?

The reason they have not been forgotten is that the irreps of those subgroups appear when restricting the irrep of $S_3$ to the fluxes within a normalizer. For instance, if we know that a dyon has flux in the conjugacy class $C_y$ and that the charge on the vertex corresponds to the 2-dim irrep $\Gamma_2^{S_3}$, we should consider the action of this irrep restricted to the elements of the normalizer $\cN_y$. One can straightforwardly check that the 2-dim irrep of the group splits into two 1-dim irreps of the subgroup $\mathbb{Z}_3$, i.e., recall Eq.~\eqref{eq:splitup_Gamma2_Ny}:
\begin{equation}
 \Gamma^{S_3}_2 |_{\cN_y} = \Gamma^{\mathbb{Z}_3}_\omega \oplus \Gamma^{\mathbb{Z}_3}_{\bar{\omega}}.
\end{equation}
Thus, the anyons $G=(C_y,\Gamma^{\mathbb{Z}_3}_\omega)$ and $H=(C_y,\Gamma^{\mathbb{Z}_3}_{\bar{\omega}})$ are accounted for. However, our Hamiltonian will give them the same mass since it does not distinguish between them. This is a general feature of our construction in the sense that the splitting of irrep of the group $G$ to recover irreps of the normalizer will happen for any group $G$.

Indeed, the statements about the correspondence between representations of the group and its subgroups can be made rigorous for any group $G$. For any finite group $G$, the $A_{\Gamma^G}$ charge projector corresponding to irrep $\Gamma^G$ will contain in its image the particle with trivial flux and $\Gamma^G$ charge, as well as all particles that have non-trivial flux $C_h$ ($h\neq e$) and their charge corresponds to the restricted representation \cite{Serre12} of $\Gamma^G$ onto the appropriate normalizer subgroup $\cN_h$:
\bq
(C_e,\Gamma^G) &\subset & \Im[A_{\Gamma^G}] \\
(C_h, \Gamma^G|_{\cN_h}) &\subset & \Im[A_{\Gamma^G}]
\eq
where $\Im[O]$ denotes the image of operator $O$ and $\subset$ means that the anyon labeled by the pair (conjugacy class, irrep) corresponds to a subspace located within the vector space on the right hand-side.

If $\Gamma^G|_{\cN_h}$ is reducible on $\cN_h$, then anyons instead will correspond to the resulting irreps:
\bq
\Gamma^G|_{\cN_h} &=& \bigoplus_{i} \Gamma_i^{\cN_h} \\
\textrm{anyon label}_i &=& (C_h,\Gamma_i^{\cN_h})
\eq
and all such anyons ($\forall i$) will have the same energy. For example, for the group $G=S_3$, anyons G and H have the same energy.

Similarly, one might ask the converse question: if we take an anyon type $(C_h, \Gamma^{\cN_h})$, does the 4-local Hamiltonian account for it? The answer is yes; one needs to consider the \emph{induced} representation $\kappa^G$ from $\Gamma^{\cN_h}$ onto the full group $G$ \cite{Serre12}. In the case that the induced representation is irreducible on $G$, then that anyon labelled $(C_h, \Gamma^{\cN_h})$ corresponds to charge $\kappa^G$
\bq
\textrm{Ind}_{\cN_h}^G(\Gamma^{\cN_h}) &=& \kappa^G \\
(C_h, \Gamma^{\cN_h}) &\subset & \Im [A_{\kappa^G}]. 
\eq
whereas, in the case the induced representation is reducible on the group $G$, then the anyon labelled $(C_h, \Gamma^{\cN_h})$ corresponds to different charge flavors $\kappa_i^G$:
\bq
\textrm{Ind}_{\cN_h}^G(\Gamma^{\cN_h}) &=& \bigoplus_i \kappa_i^G \\
(C_h, \Gamma^{\cN_h}) &\subset & \Im [A_{\kappa_i^G}]  \;\; \forall i.
\eq
For example, for $G=S_3$, anyon F is in the image of both $A_{\Gamma_1^{S_3}}$ and $A_{\Gamma_{-1}^{S_3}}$, as $\Gamma_1^{S_3}$ and $\Gamma_{-1}^{S_3}$ are the irreducible components of $\textrm{Ind}_{\cN_y}^{S_3}(\Gamma_1^{\cN_y})$.

\end{document}